%% file: swipe.tex
\documentclass{article}

\usepackage{amsmath}
\usepackage{url}
\usepackage{hyperref}
\usepackage{multirow}
\usepackage{tikz}
\usetikzlibrary{patterns}
\usepackage{pgfplots}
\usepackage{algpseudocode}
\usepackage{algorithm}
\usepackage{todonotes}
\usepackage{stmaryrd}
\usepackage{booktabs}
\usepackage{txfonts}
 \usepackage{geometry}
\newtheorem{theorem}{Theorem}
\newtheorem{proposition}{Proposition}
\newtheorem{definition}{Definition}
\newtheorem{example}{Example}
\newtheorem{proof}{Proof}

\newcommand{\schema}{\mathcal{R}}
\newcommand{\schemai}[1]{\schema_{#1}}
\newcommand{\rel}{R}
\newcommand{\attr}{\mathbb{A}}

\newcommand{\fd}{\phi}
\newcommand{\lhs}[1]{\mathsf{LHS}\left(#1\right)}
\newcommand{\rhs}[1]{\mathsf{RHS}\left(#1\right)}
\newcommand{\fdxy}[2]{#1 \rightarrow #2}
\newcommand{\fdxyt}[2]{\fdxy{\mathsf{#1}}{\mathsf{#2}}}
\newcommand{\fda}{\fdxy{X}{a}}
\newcommand{\fds}{\Phi}
\newcommand{\fdcov}{\underline{\fds}}

\newcommand{\repair}{\rel^*}
\newcommand{\repairi}[1]{\repair_{#1}}
\newcommand{\rf}[1]{\rho_{#1}}
\newcommand{\rfapp}[2]{\rf{#1}\left(#2\right)}
\newcommand{\tid}{\operatorname{tid}}
\newcommand{\seqrep}[2]{#1 \rightsquigarrow #2}
\newcommand{\fseqrep}[2]{#1 \rightsquigarrow_{F} #2}
\newcommand{\notfseqrep}[2]{\neg\left(\fseqrep{#1}{#2}\right)}
\newcommand{\cl}[1]{C_{#1}}
\newcommand{\parti}{\mathcal{P}}
\newcommand{\bags}[1]{\mathcal{B}(#1)}
\newcommand{\bproj}[2]{#1\llbracket #2\rrbracket}
\newcommand{\dsf}[1]{\mathsf{DSF}\left(#1\right)}

\newcommand{\rbefore}{\repair_{\mathsf{before}}}
\newcommand{\rafter}{\repair_{\mathsf{after}}}

\title{Cleaning data with Swipe}
\author{Toon Boeckling, Antoon Bronselaer}
\date{March 2024}

\begin{document}

\maketitle

\begin{abstract}
The repair problem for functional dependencies is the problem where an input database needs to be modified such that all functional dependencies are satisfied and the difference with the original database is minimal.
The output database is then called an \emph{optimal repair}.
If the allowed modifications are value updates, finding an optimal repair is $\mathsf{NP}$-hard.
A well-known approach to find \emph{approximations} of optimal repairs builds a Chase tree in which each internal node resolves violations of one functional dependency and leaf nodes represent repairs.
A key property of this approach is that controlling the branching factor of the Chase tree allows to control the trade-off between repair quality and computational efficiency.
In this paper, we explore an extreme variant of this idea in which the Chase tree has only one path.
To construct this path, we first create a partition of attributes such that classes can be repaired sequentially.
We repair each class only once and do so by fixing the order in which dependencies are repaired.
This principle is called \emph{priority repairing} and we provide a simple heuristic to determine priority.
The techniques for attribute partitioning and priority repair are combined in the Swipe algorithm.
An empirical study on four real-life data sets shows that Swipe is one to three orders of magnitude faster than multi-sequence Chase-based approaches, whereas the quality of repairs is comparable or better.
Moreover, a scalability analysis of the Swipe algorithm shows that Swipe scales well in terms of an increasing number of tuples.
\end{abstract}

\section{Introduction}
\label{sec:introduction}
We study the problem of repairing an inconsistent database through value modifications in the case where constraints are functional dependencies (FDs).
More precisely, starting from a database that violates some FDs (i.e., a dirty database) we seek to generate a database that satisfies all FDs (i.e., a repair) and that originates from the dirty database by sequentially changing the value of a cell in a table.
Beyond the scope of toy problems, there usually exist many possible repairs for a given dirty database and most of them are perceived as ``not good'' because they simply make too much changes to the dirty database.
To deal with this, one of the most prevalent approaches, is to define a \emph{cost model} that assigns a positive cost to each change.
Introducing this cost model gives rise to the \emph{repair problem for FDs}, which is the problem of finding a repair that minimizes the accumulative cost of all changes one must make to produce that repair.
Such a repair is called a \emph{minimal-cost repair}.
The following example illustrates the repair problem for FDs and will be used as a running example throughout the paper.

\begin{example}
\label{ex:running}
Figure~\ref{fig:example} (top) shows a snippet of the Hospital data set \cite{Xu2013}, for which the seven FDs shown must be satisfied.
If we assume that any change to an attribute value has cost $1$, then a minimal-cost repair can be obtained by changing the values marked in grey (which happen to be here the actual errors in the data set).
This repair has cost $6$ and it can be verified by the reader that no repair of cost $5$ exists.
It can be seen from this example that minimal-cost repairs are not unique, since several other repairs of cost $6$ exist.
For example, the violation of $\fdxyt{measure\ code}{condition}$ by tuples $1$ and $5$ can also be resolved by changing the value of $\mathsf{measure\ code}$ for tuple $5$.
\end{example}

\begin{figure}[!htb]
  \centering
  \includegraphics[width=0.95\textwidth]{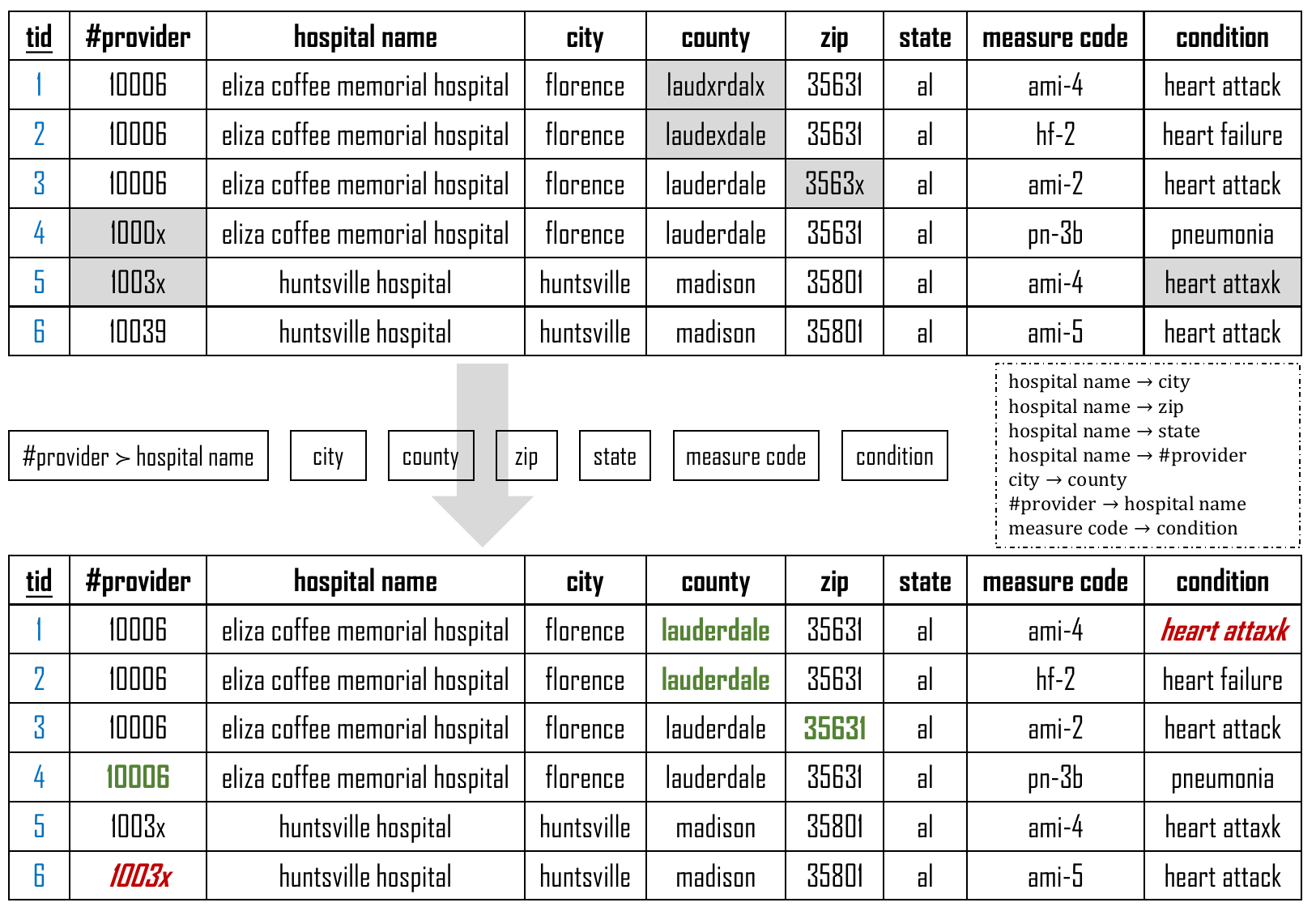}
  \caption{An example cleaning scenario with hospital data \cite{Xu2013} (top) and seven FDs (middle right).
  Actual errors in the data are marked in grey.
  A partition of attributes is shown (middle left) over which FDs are forward repairable.
  The first partition class shows a priority model over its attributes.
  For this partition, a repair obtained by using majority voting with random tie breaking as repair function is shown (bottom).
  Correct changes are shown in green bold font and incorrect changes are shown in red bold font.}
  \label{fig:example}
\end{figure}

\paragraph{The Llunatic Chase algorithm}
Although we can easily find a minimal-cost repair in Example~\ref{ex:running}, it is known that the repair problem for FDs as defined here, cannot be solved efficiently. 
Deciding whether there exists a repair with a cost lower than a constant $C$, is $\mathsf{NP}$-complete and finding constant-factor approximations of minimal-cost repairs is $\mathsf{NP}$-hard \cite{Solmaz2009}.
To deal with this, algorithms have been proposed to search for \emph{approximate} solutions \cite{Bohannon2005,Solmaz2009,Beskales2010,Livshits2018}.
In the current paper, we focus on one particular approach that relies on the Chase algorithm and has been implemented in the open-source framework Llunatic \cite{Geerts2013,Geerts2019}.
We note and emphasize here that Llunatic can in fact deal with constraints that are more expressive than FDs.
Yet, we limit the explanation of Llunatic here to FDs only, as these constraints are the focus of the current paper.
Originally, the Chase algorithm was designed to work as a proof engine for dependency implication \cite{Beeri1984,Abiteboul1995}.
Informally, the algorithm takes as input (i) a set of tuples and (ii) a set of FDs and then builds a sequence of \emph{Chase steps}.
In each such Chase step, one FD that is currently violated on the given set of tuples, becomes satisfied by equating the right-hand side attributes of the FD.
The output of a Chase step is then used as input for the next Chase step and this process continues until all FDs are satisfied.
In the original algorithm, tuples contained \emph{labelled variables} and conflict resolution among variables with different labels is done by a consistent choice (e.g., if labels are integer-valued indices, we can always choose the variable with the lower index).
In this setting, it is well-known that the Chase algorithm terminates and satisfies the \emph{Church Rosser} property \cite{Abiteboul1995}.
The latter property states that the output of the Chase algorithm does not depend on the order in which FDs are selected.

In order to use the Chase algorithm as a repair tool in the presence of constants, a modification of the Chase algorithm called the \emph{Llunatic Chase} has been proposed \cite{Geerts2019}.
In this algorithm, resolution of conflicts is based on a \emph{partial order} that models preferences of different constants and variables (e.g., null values).
If a conflict cannot be resolved by means of this partial order, special labelled variables (i.e., lluns) are introduced, which can later be used to query for human input.
While the Llunatic Chase has been shown to terminate, it does not satisfy the Church Rosser property and the outcome of a sequence of Chase steps is thus dependent on the order in which FDs are selected for repair.
To deal with this, Llunatic considers a \emph{Chase tree} in which each path is a sequence of Chase steps and leaf nodes are repairs.
From all repairs in a Chase tree, a repair can then be selected.
For example, one could consider a cost model and choose the repair with minimal cost.

Because it is infeasible to generate the complete Chase tree, Llunatic uses a \emph{cost manager} that allows the definition of a wide range of pruning strategies.
Examples of these strategies include (i) limiting the number of outgoing branches of a node by means of a branching threshold, (ii) limiting the amount of leaf nodes by a potential solutions threshold and (iii) limiting modifications to forward repairs only \cite{Geerts2019}.
The cost manager thus offers a trade-off between repair quality (i.e., generating more repairs implies a higher chance of a finding a high quality repair) and computational cost (i.e., generating more repairs implies larger Chase trees).

\paragraph{Single-path Chase trees with Swipe}
The key contribution of the current paper is to investigate an extreme case of cost management.
More specifically, we investigate a degenerate variant of the Llunatic Chase where a Chase tree represents a single sequence of Chase steps and generates only one repair.
In addition, we restrict to forward repairs only, which means that violations of an FD are always restored by equating values of the right-hand side attributes in that FD.
The main problem we face in this scenario, is how to best select the next FD to repair.
To solve this problem, we will first partition the attributes involved in a given set of FDs such that the partition is \emph{forward repairable}.
This means we can iterate over the partition classes sequentially and for each class $C$, apply forward repairing for those FDs of which the right-hand side attribute is in $C$.
Once a partition of attributes is available, the next step is to visit each partition class and for each class, consider a \emph{priority model} for the FDs we must repair for that class.
The combination of (i) the partition of attributes and (ii) the priority models for each class gives us the ability to produce a (degenerate) Chase tree with a single-sequence of Chase steps, where each Chase step uses forward repairing.
We can now summarize the key contributions of this paper as follows:
\begin{itemize}
\item We present the \emph{Swipe} algorithm to repair FD violations.
This algorithm is grounded on two key ideas: (i) the notion of a partition of attributes that is forward repairable and (ii) the notion of a priority model for FDs.
We show that for any set of FDs there always exists a partition that is forward repairable and we present a simple algorithm to construct one that is maximally refined.
In addition, we provide a simple but effective heuristic to create a priority model for FDs when a class from a forward repairable partition is given.

\item Similar to other FD repair algorithms, we use equivalence relations on the set of tuples to keep track of repair steps previously done.
In order to do this in an efficient way, the Swipe algorithm uses \emph{disjoint set forests}.
These data structures model equivalence classes in a tree-based manner and have the property that the asymptotic complexity of merging two classes, is constant.

\item We study the theoretical properties of Swipe.
We prove that it always terminates and when it does, it produces a repair.
Moreover, we show that for unary FDs (i.e., FDs with a singleton left-hand side), each FD must be repaired at most once whenever resolution of conflicts is based on choice.

\item We empirically study the trade-off between repair quality and computational cost.
On one hand, we demonstrate that the generation of repairs with Swipe is one to three orders of magnitude faster than with Llunatic.
On the other hand, the repair quality in terms of $F$-score of correctly repaired attribute values, is shown to be comparable or better.
This provides first evidence for the fact that the construction of a single sequence of Chase steps can lead to good repairs.
\end{itemize}

The remainder of this paper is organized as follows.
In Section~\ref{sec:rel-work}, we revise the vast body of literature on FD repairing and highlight the most important connections to existing methods.
In Section~\ref{ref:prelim}, we introduce the basic concepts and notations related to the relational model and functional dependencies used throughout the paper.
In Section~\ref{sec:sequential-repairing}, we first formalize the notion of an attribute partition that is (forward) repairable (Section~\ref{sec:basic-definitions}) and then develop an algorithm to produce such a partition for a given set of FDs (Section~\ref{sec:partition-building}).
Next, we show how repair of a partition class is done (Section~\ref{sec:class-repair}) and finally combine our results to formulate the Swipe algorithm (Section~\ref{sec:swipe}).
We provide an experimental analysis in Section~\ref{sec:experiments} and summarize the main contributions in Section~\ref{sec:conclusion}.

\section{Related work}
\label{sec:rel-work}
In this section, we provide a concise overview of existing solutions for repairing violations of FDs and contrast them to the proposed algorithm in this paper.
For a more detailed discussion, the interested reader is referred to overview papers \cite{Ilyas2015} and monographs \cite{Fan:2012,Ilyas2019}.
We make a distinction between (i) traditional approaches where searching repairs is cast into an \emph{(approximate) optimization problem} and (ii) more recent approaches where \emph{learning} techniques are used to generate repairs.

The traditional approach towards repairing violations of FDs is to consider a cost model that encodes a positive cost to changes made to the original data.
The goal is then to find repairs that can be produced with a minimal sum of costs.
In the setting of Consistent Query Answering (CQA) \cite{Arenas1999} one usually considers the deletion of tuples as an elementary change \cite{Wijsen2006,Livshits2018}.
In the current paper, we focus on the problem where each elementary change is the replacement of the value of one attribute in one row, with some other value.
Finding an optimal repair (i.e., a repair with minimal cost) in this setting, has been shown to be max-$\mathsf{SNP}$ hard (\cite{Solmaz2009}, Theorem 5).
Several strategies have therefore been proposed to search for approximate optimal repairs efficiently, including greedy search algorithms \cite{Bohannon2005,Bohannon2007,Cong2007,Solmaz2009}, sampling approaches \cite{Beskales2010}, usage of MAX-SAT solvers \cite{Dallachiesa2013} and alignment with knowledge graphs \cite{Chu2015}.
Many of these approaches make use of the notion of equivalence classes over the set of tuples in order to keep track of which rows must receive the same value for some attribute.
This technique has been found necessary to ensure termination of repair algorithms \cite{Bohannon2005}.
The way in which equivalence classes are used, can differ between approaches.
Equivalence classes could for example be gradually built as a repair is constructed row by row \cite{Bohannon2005}.
Another approach that has been used, initializes equivalence classes in a greedy way prior to repair, thereby determining equivalences that already exist in the dirty relation \cite{Beskales2010}.
The algorithm we present also uses equivalence relations on rows, although without this notion of greedy initialization.
Moreover, we explicitly use the notion of disjoint set forests to manage equivalence relations on rows efficiently.

When we evaluate search-based approaches, several experimental studies show that they either deal with scalability issues (even for moderate-sized data sets) \cite{Dallachiesa2013,Bohannon2005,Bohannon2007} or produce repairs of low quality \cite{Rezig2017}.
This provides empirical evidence that it is challenging to find a good trade-off between computational cost and repair quality.
One particular method that facilitates this trade-off, is the open-source framework Llunatic \cite{Geerts2013, Geerts2019}, which uses a variant of the Chase algorithm.
As mentioned in the introduction, this approach produces a Chase tree where each path is a sequence of repair steps and leaf nodes are actual repairs, associated with a certain cost.
To control the space and time complexity, the Chase tree can be pruned in several ways. 
The algorithm we present here can be thought of as a degenerate variant of such a Chase algorithm that considers a single sequence of Chase steps.
By doing so, we minimize the computational cost and try to maximize quality of repairs as much as possible.
Informally, the key idea is hereby to repair small groups of attributes of a dirty relation one after the other.
This idea is similar to the Fellegi-Holt approach to repair tuple-level constraints \cite{Fellegi1976,Boskovitz2008,Boeckling2022,BronselaerAcosta2023}.
Interestingly, in the case of tuple-level constraints, it can be shown that if the set of constraints is satisfiable, there always exists a partition composed of \emph{singleton} classes, such that we can repair sequentially. This result follows from Theorem 1 in \cite{Fellegi1976} for nominal edit rules and has been extended for other types of tuple-level constraints \cite{DeWaal2011,Boeckling2022}.
As a consequence, it is possible for these constraints to repair one attribute at a time for each separate row.
A similar property does \emph{not} hold for FDs: we cannot ensure the existence of a partition that is sequentially repairable in the sense defined here and that consists solely of singleton classes \cite{Ginsburg1982}.

Over the past decade, there has been a growing interest in using learning approaches to deal with the repair problem.
The most well-known approach is HoloClean \cite{Rekatsinas2017}, which combines integrity constraints with statistical information about a given data set to compose a probabilistic model that generates repairs.
This learning-based technique has gained a lot of attention \cite{Mahdava2020} and inspired others to increase the role of learning.
More specifically, learning approaches have also been used for error detection \cite{Mahdava2019,Heidari2019,Pham2021} where a sample of clean data is used to learn how errors should be characterized.
In these approaches, there is no need anymore to use explicit integrity constraint formalisms.
Rather, there is a model that is being trained to recognize errors.
Although these approaches are promising for the future, their experimental evaluations have shown they face scalability issues as well, displaying run times of over 10 seconds for data sets with 1000 tuples.
We believe there is a requirement for tools that allow to produce repairs much faster than that.
Evidence for that requirement can be found in a recent survey paper investigating the landscape of commercial tools for data quality monitoring \cite{Ehrlinger2022survey}.
Interestingly, none of those tools have adopted methods for automated error detection or repair, whereas we are convinced that these methods are key in measurement and improvement of data quality.

\section{Preliminaries}
\label{ref:prelim}
Let $\attr$ be a countable set of attributes.
For each $a\in \attr$, let $A$ denote the \emph{domain} of $a$.
We define a \emph{(relational) schema} as a non-empty and finite set of attributes $\schema=\{a_1,\ldots,a_k\}$ and a \emph{relation} $\rel$ with schema $\schema$ as a finite set $\rel\subseteq A_1\times \ldots \times A_k$.
Elements of $\rel$ are called \emph{tuples}.
We assume that $\schema$ always has a special attribute $\tid$ that is unique for each $r\in\rel$.
For a relation $\rel$ with schema $\schema$, a set of attributes $X\subseteq\schema$ and a predicate $P$, we denote the \emph{projection} of $\rel$ over $X$ by $\rel[X]$ and the \emph{selection} over $P$ by $\rel_{P}$.
For convenience of notation, if $X$ is a singleton set $\{a\}$, we denote $\rel[\{a\}]$ by $\rel[a]$.
In a similar way, if $R$ contains only one tuple $r$, we use the notation $r[X]$.
Note that the projection $\rel[X]$ is a set and therefore contains no duplicate tuples.
In some cases, it is useful to restrict the schema to attributes $X$ but keep duplicates.
We denote $\bproj{\rel}{X}$ as the \emph{multiset} of values obtained after restricting the schema to attributes $X$ without removal of duplicates.

A \emph{functional dependency} (FD) $\fd$ defined over a schema $\schema$ is an expression of the form $X\rightarrow a$ such that $X\subseteq \schema$ and $a\in \schema$.
In this expression, $X$ is called the left-hand side and $a$ is called the right-hand side of the FD\footnote{We can assume the right-hand side is a single attribute without loss of generality.}.
For an FD $\fd$, we will use the notation $\lhs{\fd}$ for the left-hand side of $\fd$ and $\rhs{\fd}$ for the right-hand side of $\fd$.
A relation $\rel$ with schema $\schema$ \emph{satisfies} $\fd$ (denoted by $\rel\models \fd$) if:
\[\forall r_1\in \rel: \forall r_2\in \rel: r_1[X]=r_2[X] \Rightarrow r_1[a]=r_2[a].\]
Similarly, $\rel$ \emph{satisfies} a set of FDs $\fds$ (denoted by $\rel\models \fds$) if it satisfies all $\fd\in\fds$.
For a set of FDs $\fds$ and a set $Z\subseteq\schema$, the \emph{projection} of $\fds$ over $Z$ is denoted by $\fds[Z]$ and defined as:
\[\fds[Z] = \{\fd\mid \fd\in\fds\wedge \lhs{\fd}\subseteq Z\wedge \rhs{\fd}\in Z\}.\]
An FD $\fd$ defined over $\schema$ is said to by \emph{implied} by a set of FDs $\fds$, if for any relation $\rel$ with schema $\schema$, we have $\rel\models\fds \Rightarrow \rel\models\fd$ (i.e., satisfaction of $\fds$ implies satisfaction of $\fd$).
We use the notation $\fds\models\fd$ to say that $\fd$ is logically implied by the set $\fds$.
Two sets of FDs $\fds$ and $\fds'$ are said to be \emph{equivalent}, denoted by $\fds\equiv\fds'$ if:
\[\left(\forall \fd\in\fds: \fds'\models \fd\right) \wedge \left(\forall \fd'\in\fds': \fds\models \fd'\right).\]
An FD $\fda$ is called \emph{minimal} if it is not implied by any $\{\fdxy{X'}{a}\}$ where $X'\subset X$.
Finally, for a given set of FDs $\fds$, a \emph{minimal cover} of $\fds$ is any set $\fdcov$ such that (i) $\fds\equiv\fdcov$, (ii) all $\fd\in\fdcov$ are minimal and (iii) no proper subset of $\fdcov$ is equivalent with $\fds$.

\begin{table}[ht!]
\caption{Overview of symbols and key concepts.}
\label{table:symbols-overview}
\centering
\begin{tabular}{ll}

\toprule
\textbf{Symbol} & \textbf{Meaning} \\
\midrule
$\attr$        & Countable set of attributes \\
$a\in\attr$    & An attribute with domain $A$\\
$\schema$ 	   & Relational schema $\{a_1,\ldots, a_k\}$\\
$\tid$         & Tuple identifier attribute\\
$\rel$		   & Relation with schema $\schema$\\
$r\in\rel$ 	   & A tuple in relation $\rel$\\
$\rel[X]$      & Projection of $\rel$ over $X\subseteq\schema$\\
$\bproj{\rel}{X}$ & Projection of $\rel$ over $X\subseteq\schema$ with preservation of duplicates\\
$\fd$          & Functional dependency (FD) of the form $X\rightarrow a$ \\
$\fds$         & A set of FDs \\
$\fds\models\fd$& FD $\fd$ is implied by $\fds$ \\
$\fdcov$       & A minimal cover of $\fds$ \\
$\fds[Z]$      & Projection of $\fds$ over $Z$ \\
\midrule
$\parti$                 & A partition of $\schema$ of the form $\left[\cl{1}\ldots\cl{m}\right]$ \\
$\schemai{i}$            & For $\parti=\left[\cl{1}\ldots\cl{m}\right]$ and $1\leq i\leq m$, $\schemai{i}=\{\cl{1}\cup\ldots\cup\cl{i}\}$\\
$\repairi{i}$            & Partial repair: a relation with schema $\schema$ that satisfies $\fds[\schemai{i}]$\\
$\seqrep{\parti}{\fds}$  & $\fds$ is sequentially repairable over $\parti$ \\
$\fseqrep{\parti}{\fds}$ & $\fds$ is forward repairable over $\parti$ \\
$P^+$                    & The preorder obtained from a set of FDs $\fds$\\
\midrule
$\dsf{a}$                & A disjoint set forest (DSF) for $a$\\
$\rf{a}$                 & A repair function for $a$\\
\bottomrule
\end{tabular}
\end{table}

\section{Sequential repairing}
\label{sec:sequential-repairing}

\subsection{Basic definitions}
\label{sec:basic-definitions}
Consider a schema $\schema$ and a set of FDs $\fds$ defined over $\schema$.
Let $\parti = \left[\cl{1}\ldots\cl{m}\right]$ denote a \emph{partition} of $\schema$ where $m\leq k$ and for any $i\in\{1,\ldots,m\}$, $\cl{i}\subseteq\schema$ are mutually disjoint \emph{partition classes}.
We assume in this paper that partition classes $\cl{i}$ are \emph{ordered} by their index $i$ and we refer to this order as the \emph{natural} order of the partition.
In the context of such a partition $\parti$, for any $i\in\{1,\ldots,m\}$, we denote by $\repairi{i}$ a relation with schema $\schema$ such that $\repairi{i} \models \fds[\schemai{i}]$ where $\schemai{i} = \{\cl{1}\cup\ldots\cup \cl{i}\}$.
In the context of our repairing approach, we will call such a relation a \emph{partial repair}.
We now say that $\fds$ is \emph{sequentially repairable} over $\parti$, denoted by $\seqrep{\parti}{\fds}$, if:
\begin{equation}
\forall i\in\{1,\ldots,m\}:\forall \repairi{i-1}:\exists \repairi{i}:\repairi{i}[\schemai{i-1}] = \repairi{i-1}[\schemai{i-1}]
\end{equation}
where we adopt the convention that $\repairi{0}=\rel$.
In words, if $\fds$ is sequentially repairable over $\parti$ then any partial repair $\repairi{i-1}$ can be transformed into a partial repair $\repairi{i}$ by modifying only the values for attributes in $\cl{i}$.
By repetition, we eventually obtain a relation $\repairi{m}$ that satisfies $\fds$.
The relevance of sequential repairability lies in the fact that if $\seqrep{\parti}{\fds}$, then classes from $\parti$ can be repaired \emph{one at a time} in order of increasing $i$.
We now make two refinements to the notion of sequential repairability in order to make it a useful repair instrument.

First, it easy to see that for any schema $\schema$ and any set $\fds$ defined over $\schema$, there always exists at least one (trivial) partition for which $\fds$ is sequentially repairable.
Indeed, since FDs are always satisfiable, we have that $\seqrep{\left[\schema\right]}{\fds}$.
This partition is not very useful because it simply tells us we can repair violations of $\fds$.
However, if we are able to \emph{refine} this partition (i.e., splitting classes into disjoint sub-classes) and still satisfy the condition of sequential repairability, we can separate the treatment of different FDs and exploit this property in a repair algorithm.
For that reason, we are interested here in those partitions $\parti$ that are \emph{maximally refined} without losing the property of sequentially repairability for $\fds$.

Second, if $\seqrep{\parti}{\fds}$, it is possible that FDs must be satisfied by either equating the values for the right-hand side attribute or differentiating the values for the left-hand side attributes.
For example, if $\schema=\{a,b\}$ and $\fds=\{\fdxy{a}{b}\}$ then we have $\seqrep{\left[\{a\},\{b\}\right]}{\fds}$ as well as $\seqrep{\left[\{b\},\{a\}\right]}{\fds}$.
That is, for given values of $a$, we can choose values for $b$ and make sure equal values for $a$ imply equal values for $b$.
Similarly, for given values of $b$, we can choose values for $a$ that do not cause a violation of $\fdxy{a}{b}$.
In this paper, we restrict repairing to the case where violations of FDs are resolved by changing the values of the right-hand side attribute.
This operation is also known as \emph{forward} repairing \cite{Geerts2019}.
We provide the following definition.
\begin{definition}
\label{def:forward-seq-repair}
For a schema $\schema$ and a set of FDs $\fds$ defined over $\schema$, let $\parti$ be a partition of $\schema$.
Then $\fds$ is \emph{forward repairable} over partition $\parti$, denoted by $\fseqrep{\parti}{\fds}$, if and only if:
\begin{equation}
\label{eq:forward-repair}
\forall \cl{i}\in\parti:\forall \fd\in \fds[\schemai{i}] \setminus \fds[\schemai{i-1}]: \rhs{\fd} \in \cl{i}
\end{equation}
\end{definition}
In words, $\fds$ is forward repairable if for every class $\cl{i}$, every FD that was not considered in the scope of previous classes, has a right-hand side attribute that is an element of $\cl{i}$.
We now have the following result.
\begin{proposition}
\label{prop:seq-rep-implication}
For a schema $\schema$ and FDs $\fds$ defined over $\schema$, we have $\left(\fseqrep{\parti}{\fds}\right) \Rightarrow \left(\seqrep{\parti}{\fds}\right)$.
\end{proposition}
In this paper, we seek to repair violations of FDs by first building a partition of attributes that allows forward repairing and then apply a repair algorithm to each class of this partition, following the natural order of that partition.
This repair algorithm uses a \emph{priority} model that dictates the order in which FDs are repaired.
We now demonstrate the main ideas behind this algorithm in the context of the running example.
\begin{example}
\label{ex:process}
In Figure~\ref{fig:example}, the attribute partition (middle left) is forward repairable for the FDs.
This partition is also maximally refined w.r.t. forward repairability.
That means we cannot split classes without breaking the condition of forward repairability.
To construct a repair $\repair$, we visit the classes of $\parti$ in their natural order.
The first class is $\cl{1} = \{\mathsf{hospital\ name}, \mathsf{\# provider}\}$ and we have:
\[\fds[\cl{1}] = \{\fdxyt{hospital\ name}{\#provider}, \fdxyt{\#provider}{hospital\ name}\}.\]
We then sort these FDs (this sorting is made more precise later in the paper) and visit them in order.
Suppose for now we first visit $\fdxyt{hospital\ name}{\#provider}$, then we find two violations: one for tuples $\{1,2,3,4\}$ and one for tuples $\{5,6\}$.
We then do two things.
First, we register that these two groups of tuples must have the same value for $\mathsf{\#provider}$ in a special data structure.
This data structure is used to keep track of decisions made during repair of visited FDs.
Second, we resolve the violation by applying a function that maps the bag of values observed for $\mathsf{\#provider}$ in the violating tuples to a single value.
In this example, we do this by majority voting with random tie breaking.
Concretely, this means we change $\mathsf{\#provider}$ for $\tid=4$ into $\mathsf{10006}$.
For the other violation, we have a random choice between the values so assume we choose $\mathsf{1003x}$.
We then visit $\fdxyt{\#provider}{hospital\ name}$ and find no violations.
Since all FDs in $\fds[\cl{1}]$ are now satisfied, we have constructed the first partial repair $\repairi{1}$.
We then proceed to $\cl{2}=\{\mathsf{city}\}$. 
The set $\fds[\cl{1}\cup\cl{2}]$ now has one additional FD, which is $\fdxyt{hospital\ name}{city}$ and it has no violations.
Moreover, the FDs from the previous step are ensured to remain satisfied so we now have a second partial repair $\repairi{2}$ for which $\repairi{2} \models\fds[\cl{1}\cup\cl{2}]$.
This procedure continues until we have visited $\cl{7}=\{\mathsf{condition}\}$.
After this last step, the original data has been iteratively transformed into a repair $\repair$ that satisfies all FDs.
An example repair under the assumption of majority voting with random tie breaking to fix violations, is shown in Figure~\ref{fig:example} (bottom).
\end{example}
From the example given above, two important questions arise.
The first question deals with the generation of a maximally refined partition that is forward repairable for a given set of FDs.
We will treat this question in Section~\ref{sec:partition-building}.
The second question deals with how violations of FDs are resolved for a single class.
We will treat this question in Section~\ref{sec:class-repair}.
After that, we combine all these ideas in Section~\ref{sec:swipe} to present the Swipe algorithm.

\subsection{Attribute partition building}
\label{sec:partition-building}
This section presents a method to construct, for given $\schema$ and $\fds$ over $\schema$, a partition $\parti$ that (i) is maximally refined and (ii) satisfies $\fseqrep{\parti}{\fds}$.
To do so, note that Definition~\ref{def:forward-seq-repair} (Eq.~\eqref{eq:forward-repair}) implies that for any FD $\fdxy{X}{a}$ in $\fds$, attributes in $X$ must be part of a class that does \emph{not lie after} the class that contains $a$.
We will encode this information in a preorder relation constructed from $\fds$.

As a first step, we verify that $\fds$ does not contain any trivial or implied FDs.
This is to avoid that trivial FDs introduce unnecessary constraints on the order of attributes.
In the scope of the running example in Figure~\ref{fig:example}, introducing the trivial FD $\fdxyt{\{county,zip\}}{county}$ would create a constraint on the order of the classes that contain $\mathsf{county}$ and $\mathsf{zip}$.
This constraint is unnecessary since the FD is always satisfied.
The same observation holds for FDs that are not minimal (e.g., $\fdxyt{\{hospital\ name, state\}}{zip}$).
To ensure that we encode only necessary requirements, we transform $\fds$ into a \emph{minimal cover} $\fdcov$.
From such a minimal cover, we compose a binary relation $P\subseteq\schema\times\schema$ that satisfies the following criteria:
\begin{enumerate}
\item $\forall a\in\schema: (a,a)\in P$,
\item $\forall \left(\fda\right)\in \fdcov: \forall b\in X: (b,a) \in P$.
\end{enumerate}
Here, $(b,a)\in P$ expresses the requirement that $b$ should not occur after $a$ in the sequence of partition classes.
Next, $P$ is transformed into its \emph{transitive closure} denoted by $P^+$ \cite{Warren75}.
The obtained relation $P^+$ is the finest preorder (reflexive and transitive) that contains $P$.
We say that $P^+$ is the preorder obtained from $\fds$.
This preorder contains an equivalence relation $\equiv_{P^+}$ on $\schema$.
More specifically, we have $a\equiv_{P^+} a'$ whenever both $(a,a')\in P^+$ and $(a',a)\in P^+$.
Consider now any partition $\parti=\left[\cl{1}\ldots\cl{m}\right]$ such that:
\begin{equation}
\label{eq:equiv-criterion}
\forall a\in\schema:\forall a'\in\schema: a\equiv_{P^+} a' \Leftrightarrow \left(\exists \cl{i}\in\parti:a\in\cl{i}\wedge a'\in\cl{i}\right)
\end{equation}
and:
\begin{equation}
\label{eq:sort-criterion}
%\forall \cl{i}\in\parti: \forall \cl{j}\in\parti:i\neq j\Rightarrow \left( \left(\exists a\in\cl{i}:\exists a'\in\cl{j}:(a,a')\in P^+ \right)\Rightarrow i<j\right)
\forall \cl{i}\in\parti: \forall \cl{j}\in\parti:i < j \Rightarrow \left( \forall a\in\cl{i}:\forall a'\in\cl{j}:(a',a)\notin P^+ \right)
\end{equation}
A partition $\parti$ that satisfies these two constraints is said to be induced by $P^+$.
A partition induced by $P^+$ can be constructed readily by first assigning each equivalence class from $\equiv_{P^+}$ to a class $\cl{i}$ and then sort the obtained classes to enforce Eq.\eqref{eq:sort-criterion}.
This construction shows that a partition induced by $P^+$ is a topological sort of the quotient set $\schema/\equiv_{P^+}$.
We can now state the following result.
\begin{theorem}
\label{theorem:partition}
For a schema $\schema$ and a set of FDs $\fds$ defined over $\schema$, let $\parti$ be a partition induced by $P^+$.
Then (i) $\fseqrep{\parti}{\fds}$ and (ii) $\notfseqrep{\parti'}{\fds}$ for any refinement $\parti'$ of $\parti$.
\end{theorem}
\begin{figure}[!htb]
  \centering
  \includegraphics[width=0.70\columnwidth]{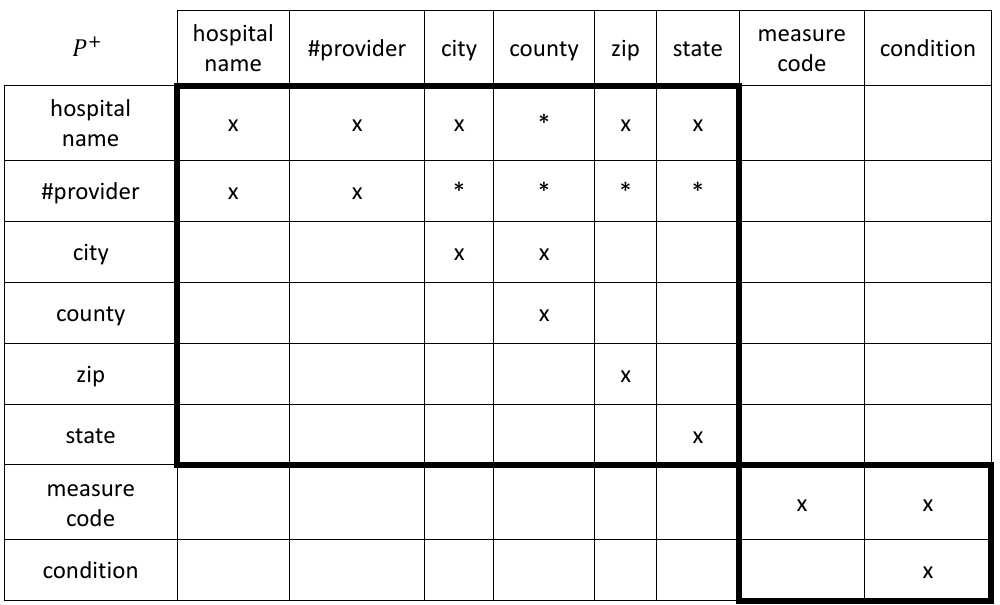}
  \caption{The construction of a preorder $P^+$ for the FDs from Figure~\ref{fig:example} (middle).
  Elements derived directly from the FDs are marked with `x' and elements added to compute the transitive closure are marked with `*'.
  The preorder is the union of two weak orders, marked by bold lines.}
  \label{fig:preorder}
\end{figure}

\begin{example}
For the FDs shown in Figure~\ref{fig:example} (middle), Figure~\ref{fig:preorder} shows the construction of the preorder $P^+$.
First, we derive $P$ from a minimal cover of the given FDs as explained.
In our example, the available FDs are already a minimal cover and thus we proceed with the FDs as given.
In Figure~\ref{fig:preorder}, the entries in $P$ are marked by `x' symbols.
We then compute the transitive closure $P^+$ and the additional entries to get $P^+$ are marked by `*' symbols.
Because $\mathsf{hospital\ name}\equiv_{P^+}\mathsf{\# provider}$, both attributes must belong to the same partition class.
Attributes determined by $\mathsf{hospital\ name}$ appear in singleton classes after the class that contains $\mathsf{hospital\ name}$ and $\mathsf{\# provider}$.
Finally, $\mathsf{condition}$ should be after $\mathsf{measure\ code}$.
\end{example}

\subsection{Priority repair}
\label{sec:class-repair}
Assume a partition $\parti=[\cl{1}\ldots\cl{m}]$ for $\schema$ such that $\fseqrep{\parti}{\fds}$ for a given set of FDs $\fds$.
Definition~\ref{def:forward-seq-repair} then ensures we can visit classes $\cl{i}$ in their natural order and generate a partial repair $\repairi{i}$ by changing only attributes from $\cl{i}$.
That is, we are given a partial repair $\repairi{i-1}$ and we must modify attributes $\cl{i}$ to fix violations of FDs $\fds[\schema_i]$ and obtain a new partial repair $\repairi{i}$.
To do this, we propose a technique called \emph{priority repairing}.
The main principles of this technique are the following:
\begin{itemize}
\item FDs are inspected and repaired in a specific order, dictated by a \emph{priority} model.
The key idea is to repair FDs that contain less reliable attributes in their right-hand side first and thereby try to maximize the accuracy of the repairs.
We propose a simple heuristic to estimate the reliability of an attribute by means of the estimated number of changes that an attribute requires.

\item In order to keep track of tuples that must receive the same value for $a\in\cl{i}$, we use a \emph{disjoint set forest} (DSF) \cite{Galler1964} to represent equivalence classes on tuples.
A DSF has the advantage that the necessary operations on equivalence classes have a near-constant time complexity \cite{Tarjan1975}.
This mitigates the complexity of FD \emph{revision}, where one FD must be repaired again because of possible new violations caused by the repair of another FD.

\item The value assigned to an equivalence class for $a\in\cl{i}$ is determined by a \emph{repair function} $\rf{a}$.
This is a function that maps the bag of values of $a$ involved in a violation onto a new value.
We have a particular interest in those repair functions that are \emph{preservative}.
These functions simply choose one of the values involved in a violation.

\end{itemize}
In the remainder of the section, we detail and formalize these main principles.

\paragraph{Priority model} In order to repair class $\cl{i}$, we must ensure that the FDs over $\fds[\schemai{i}]$ are satisfied.
Clearly, any FD from $\fds[\schemai{i-1}]$ contains no attributes from $\cl{i}$ and must therefore not be considered.
We can therefore restrict ourselves to the set $\fds_i:=\fds[\schemai{i}]\setminus\fds[\schemai{i-1}]$.
In this set there can be FDs $\fd$ for which:
\begin{equation}
\label{eq:pilot-fd}
\lhs{\fd}\cap\cl{i}=\emptyset.
\end{equation}
Those FDs have a left-hand side that contains only attributes that are clean at the time of repairing class $\cl{i}$ because these attributes are all part of a class $\cl{j}$ with $j<i$.
Such FDs are called \emph{pilot} FDs for class $\cl{i}$ and we repair them first, with two clear advantages.
First, because we rely on \emph{clean} data to make the first changes to attributes $a\in\cl{i}$, we expect these changes to be accurate.
Second, pilot FDs never require revision as the values of their left-hand side attributes never change.
Therefore, they need to be considered only once.
\begin{example}
In Example~\ref{ex:process}, there are no pilot FDs for class $\cl{1}$.
For all other classes, we have only pilot FDs because the left-hand side of the FDs contains no attributes that are part of the class.
\end{example}
If an FD is not a pilot FD for class $\cl{i}$, it contains at least one attribute from $\cl{i}$ in the left-hand side.
Those FDs are further sorted by means of a simple \emph{priority model} that encodes the order in which attributes must be repaired. 
This priority model is a simple total order $\succ$ on attributes $\cl{i}$ where $a \succ a'$ indicates that $a$ must be repaired prior to $a'$.
In other words, we require that $\fdxy{X}{a}$ is repaired \emph{before} $\fdxy{X'}{a'}$ whenever $a \succ a'$.
To construct a priority model, we propose to rank attributes within class $\cl{i}$ by increasing reliability.
To make such a ranking, the most obvious method would be to ask for input of a human supervisor, which is not uncommon in repair approaches \cite{Geerts2013,Geerts2019}.
However, such input might not be easy to collect, either because there is no supervisor, or a supervisor cannot easily rank attributes.
In that case, we propose to use an \emph{estimate} of the number of changes we will need to make to values of $a$ as follows.

Suppose we have an FD $\fd:=\fda$ and a relation $\rel$, then for any $x\in\rel[X]$, the bag $\bproj{\rel_{X=x}}{a}$ gives us the values for attribute $a$ of those tuples for which $X=x$.
If there are \emph{different} elements in this bag, we must change values for $a$ in order to satisfy $\fda$.
If we denote the element with \emph{highest} multiplicity\footnote{If there are multiple such values, we choose one at random.} in this bag by $\mathsf{mv}\left(x,a\right)$, then changing the values of $a$ for tuples $\rel_{X=x\wedge a\neq \mathsf{mv}\left(x,a\right)}$ allows to resolve the conflict with a \emph{minimal} amount of changes.
We can now collect these tuples for all values $x\in\rel[X]$:
\begin{equation}
\mathsf{VioFD}\left(\fda\right):=\bigcup_{x\in\rel[X]} \rel_{X=x\wedge a\neq \mathsf{mv}\left(x,a\right)},
\end{equation}
and finally aggregate over all FDs $\fd\in\fdcov$ where $\rhs{\fd} = a$:
\begin{equation}
\mathsf{Vio}\left(a\right):=\bigcup_{\fd\in\fdcov\wedge \rhs{\fd} = a} \mathsf{VioFD}\left(\fda\right)
\end{equation}
The set $\mathsf{Vio}\left(a\right)$ basically gives us the tuples in $\rel$ that require a change to attribute $a$ if (i) we are interested in keeping the number of changes as small as possible and (ii) we would consider each FD independently.
The more changes an attribute requires, the less reliable that attribute is.
We therefore use $|\mathsf{Vio}\left(a\right)|$ to build our priority model such that:
\begin{equation}
a \succ a' \Leftrightarrow |\mathsf{Vio}\left(a\right)| \geq |\mathsf{Vio}\left(a'\right)|
\end{equation}
The rationale of this heuristic is that since we now must sort non-pilot FDs, there are attributes from $\cl{i}$ that appear in the left-hand side of these FDs.
It is likely that some of these attributes will themselves contain errors (i.e., because they appear in the right-hand side of another FD).
We thus seek to first repair an FD that has attributes in the left-hand side, that are \emph{maximally} reliable.
We now make two observations about this heuristic.

First, the assumption about FD independence, is a naive one.
When repairing two FDs with the same right-hand side attribute $a$ we must merge equivalence classes induced by both FDs (this is explained later) and our estimate neglects this.
We accept this error in our estimate at the benefit of simple and efficient computation.
Moreover, we will show empirically that classes $\cl{i}$ are in practical scenarios often singleton set, in which case the estimate must not be computed because all FDs are then by definition pilot FDs.
If we do have multiple attributes in $\cl{i}$, then empirical results show that the heuristic leads to sequences of FDs with repair quality superior to other possible sequences.

A second observation is that $|\mathsf{Vio}\left(a\right)|$ can be computed at different times with different results.
One could for example compute $|\mathsf{Vio}\left(a\right)|$ on the original relation $\rel$, before any modifications are done.
The advantage of doing so would be that estimates are unbiased and not dependent on changes made by the Swipe algorithm.
An alternative strategy is to postpone the computation of $\mathsf{Vio}\left(a\right)$ up to the moment where the class that contains $a$ is visited.
In other words, we then compute $\mathsf{Vio}\left(a\right)$ from the partial repair $\repairi{i-1}$.
Variations of this strategy are possible.
We can for example estimate reliability before repair of $\cl{i}$ starts or postpone it to the point where pilot FDs have been repaired.
Empirically, we have however not observed any significant differences between these variations on the data sets we tested.
In the remainder of this paper, we favor a strategy where estimation of reliability is done on \emph{partial repairs}.
\begin{example}
\label{ex:sorting}
In Example~\ref{ex:process}, we have $\cl{1}=\{\mathsf{hospital\ name}, \mathsf{\#provider}\}$.
Both attributes appear once in the right-hand side of an FD.
It can be verified that $\mathsf{Vio}\left(\mathsf{hospital\ name}\right)=\emptyset$ and $\mathsf{Vio}\left(\mathsf{\#provider}\right)$ is either $\{4,5\}$ or $\{4,6\}$.
It follows that $\mathsf{\#provider} \succ \mathsf{hospital\ name}$ and thus that we repair $\fdxyt{hospital\ name}{\#provider}$ first and then $\fdxyt{\#provider}{hospital\ name}$.
\end{example}
Example~\ref{ex:sorting} demonstrates (i) the estimation of reliability of attributes and (ii) the implied order on FDs.
Note that the two sorting principles we use (i.e., pilot FDs first and then least reliable right-hand side first) provides us with a \emph{partial} order on the FDs we must repair.
Any total order that \emph{contains} this partial order can be used in Swipe.

\paragraph{Tuple equivalence}
Having determined the order in which FDs are treated, we turn our attention to the actual repair process.
In order to repair attributes in $\cl{i}$, we consider for each $a\in\cl{i}$, an equivalence relation on the set of tuples from $\repairi{i-1}$.
We use these equivalence classes to keep track of which tuples must receive the same value for $a$ in order to satisfy FDs $\fd$ for which $\rhs{\fd} = a$.
More specifically, whenever two tuples are in the same equivalence class for attribute $a$, we require they must receive the same value for $a$.
Initially, each tuple is put in a separate singleton class.
As we fix violations of FDs of the form $\fda$, we require that tuples with the same value for $X$ also have the same value for $a$.
In other words, if we observe tuples in $\repairi{i-1}$ with equal values for $X$, we must \emph{merge} the corresponding classes in which these tuples appear.
In order to do this efficiently, we use a \emph{disjoint set forest} \cite{Galler1964} for each attribute $a$, denoted by $\dsf{a}$.
With this data structure, each equivalence class is represented by a \emph{tree} on tuple identifiers (i.e., values for $\tid$).
Each $\dsf{a}$ basically supports three operations: 
\begin{itemize}
\item Operation $\mathsf{makeset}\left(\mathsf{id}\right)$ adds a new singleton set (i.e., $\{\mathsf{id}\}$) to the forest.
Such a singleton set is represented by a one-node tree.
This operation is used to initialize each $\dsf{a}$.
\item Operation $\mathsf{find}\left(\mathsf{id}\right)$ finds the root $\tid$ of the tree in which $\mathsf{id}$ resides and is used to determine if two values are in the same class or not.
\item Operation $\mathsf{union}\left(\mathsf{id}_1, \mathsf{id}_2\right)$ merges the classes in which $\mathsf{id}_1$ and $\mathsf{id}_2$ reside.
This is done by finding the root nodes of the trees and in case the root nodes are distinct, attach one root as a child of the other.
During the repair of an FD $\fda$, this operation is used to merge equivalence classes for tuples that have equal values for $X$.
\end{itemize}
A key result is that the three elementary operations of a disjoint set forest can be implemented such that their \emph{asymptotic} time complexity is very efficient.
More specifically, one can provide an implementation such that the asymptotic time complexity of  $\mathsf{makeset}\left(\mathsf{id}\right)$ and $\mathsf{union}\left(\mathsf{id}_1, \mathsf{id}_2\right)$ is $\mathcal{O}\left(1\right)$, while that of $\mathsf{find}\left(\mathsf{id}\right)$ is $\mathcal{O}\left(\alpha(n)\right)$ where $\alpha(.)$ is an inverse functional of the Ackermann function and $n$ is the number of elements in the class where $\mathsf{id}$ is located \cite{Tarjan1975,Tarjan1984,BenAmram2011}.
The asymptotic efficiency of these operations is key because it allows us to provide an efficient strategy for the update of $\dsf{a}$ during the repair of an FD $\fda$.
More specifically, if we want to repair $\fda$, then we must update $\dsf{a}$ such that after the update we have:
\begin{equation}
\forall r\in\repairi{i-1}:\forall r'\in\repairi{i-1}:r[X]=r'[X]\Rightarrow \mathsf{find}\left(r[\tid]\right) = \mathsf{find}\left(r'[\tid]\right)
\end{equation}
In words, after the update is done, we want that any two tuples from $\repairi{i-1}$ with equal values for $X$ are equivalent in $\dsf{a}$ and thus must receive the same value for $a$.
\begin{algorithm}[!htb]
\caption{Update $\dsf{a}$ for $\fda$ and $\repairi{i-1}$}
\label{algo:update-dsf}
\begin{algorithmic}[1]
\Procedure{update}{$\repairi{i-1}, \fda, \dsf{a}$}
\State $H \gets []$ \Comment{Initialize empty map} \label{line:init-map}
\For{$r\in\repairi{i-1}$}
    \If{$H[r[X]] = \textsc{null}$}
        \State $H[r[X]] \gets \dsf{a}.\mathsf{find}\left(r[\tid]\right)$ \label{line:register} \Comment{Register root}
    \Else
        \State $\dsf{a}.\mathsf{union}\left(r[\tid], H[r[X]]\right)$ \label{line:merge} \Comment{Merge classes}
        \State $H[r[X]] \gets \dsf{a}.\mathsf{find}\left(r[\tid]\right)$ \label{line:update-map}
    \EndIf
\EndFor
\EndProcedure
\end{algorithmic}
\end{algorithm}

The pseudo-code shown in Algorithm~\ref{algo:update-dsf} provides a simple way to do so by using a map $H$ with constant-time complexity for $\mathsf{get}$ and $\mathsf{put}$ operations.
This map keeps, for each value in $\repairi{i-1}[X]$, the root node of the tree in $\dsf{a}$ in which tuples with this value reside.
We then iterate over tuples $r\in\repairi{i-1}$ and check if $H$ currently contains $r[X]$.
If this is not the case, we register the root node of the tree where $r[\tid]$ is currently located (line~\ref{line:register}).
Else, we merge the current root node with the tree where $r[\tid]$ is currently located and update the map $H$ (lines~\ref{line:merge} - \ref{line:update-map}).
The update procedure described here requires in each step one $\mathsf{find}$ operation and possibly one $\mathsf{union}$ operation on $\dsf{a}$, which implies an asymptotic complexity of $\mathcal{O}\left(|\repairi{i-1}|\cdot\alpha\left(n\right)\right)$, with $n$ the size of the largest equivalence class in $\dsf{a}$.
\begin{example}
\label{ex:update-dsf}
In Example~\ref{ex:sorting} it is shown that during the visit of $\cl{1}$, we first repair $\fdxyt{hospital\ name}{\#provider}$.
This is the first FD we visit, which means $\dsf{\mathsf{\#provider}}$ contains a singleton class for each tuple.
Then application of
\[\mathbf{update}\left(\rel, \fdxyt{hospital\ name}{\#provider}, \dsf{\mathsf{\#provider}}\right)\]
modifies $\dsf{\mathsf{\#provider}}$ in such a way that tuples with equal values for $\mathsf{hospital\ name}$ are in the same class.
In other words, after the update, $\dsf{\mathsf{\#provider}}$ has classes $\{\{1,2,3,4\},\{5,6\}\}$.
\end{example}

\paragraph{Fixing violations with repair functions} 
After updating $\dsf{a}$ for some FD $\fda$ with Algorithm~\ref{algo:update-dsf}, $\dsf{a}$ represents classes of tuple identifiers from $\repairi{i-1}$ that must have the same value for $a$.
In other words, $\dsf{a}$ corresponds with set $\{\mathcal{E}_1,\ldots,\mathcal{E}_\ell\}$ such that $\forall i\in\{1,\ldots,\ell\}: \mathcal{E}_i\subseteq \repairi{i-1}[\tid]$.
For each such set $\mathcal{E}$, we can get the corresponding tuples from $\repairi{i-1}$ and verify whether or not these tuples have equal values for $a$.
If not, the tuples violate $\fda$ and we fix this violation by application of a repair function  defined by $\rf{a}: \bags{A} \rightarrow A$, where $\bags{A}$ is the set of all \emph{bags} (or \emph{multisets}) with elements from $A$.
A repair function maps a bag of values from $A$ onto a single value\footnote{In the literature on repairing FDs, many existing approaches allow labelled variables during repair, but we restrict ourselves to constant values only.} from the domain $A$.
A particular class of repair functions are those that are \emph{preservative}, which means that $\rfapp{a}{B}\in B$ for any bag $B\in\bags{A}$ \cite{Bronselaer2010}.
Our interest in preservative repair functions is not coincidental.
First of all, preservative repair functions possess important properties: they are \emph{idempotent} ($\rho\{v,\ldots,v\}=v$), they ensure that for any repair $\repair$, we have $\forall a\in \schema: \repair[a]\subseteq \rel[a]$ (i.e., they do not introduce new values in the repair) and they play an important role in notion of cardinality-minimal repairs \cite{Beskales2010}.
Second, preservative repair functions have some well-known representatives like majority voting, weighted voting and median or min/max selection for domains equipped with an order relation.
Finally, we will show in the following that they play a key role in the avoidance of revisions for unary FDs.

\begin{algorithm}[!htb]
\caption{Fix violations of $\fda$}
\label{algo:fix}
\begin{algorithmic}[1]
\Procedure{fix}{$\repairi{i-1}, \fda, \dsf{a}$}
\State $\textbf{update}\left(\repairi{i-1}, \fda, \dsf{a} \right)$ \Comment{Algorithm~\ref{algo:update-dsf}} \label{line:update-equi}
\State $\mathsf{fixes}\gets 0$ \Comment{Count violations}
\For{$\mathcal{E}\in \dsf{a}$} \Comment{For each class in $\dsf{a}$}
\State $\repair_\mathcal{E}\gets \{r\mid r\in\repairi{i-1}\wedge r[\tid] \in \mathcal{E}\}$ \label{line:get-tuples}
\If{$|\repair_\mathcal{E}[a]| > 1$} \Comment{Violation check} \label{line:violation}
    \State $v_\mathsf{fix}\gets \rf{a}\left(\bproj{\repair_{\mathcal{E}}}{a}\right)$ \label{line:repair-function}\Comment{Apply repair function}
    \State \textbf{exec update} $\repairi{i-1}$ \textbf{set} $a=v_\mathsf{fix}$ \textbf{where} $\tid\in\mathcal{E}$ \label{line:fix}\Comment{Fix violation}
    \State $\mathsf{fixes}$++
\EndIf
\EndFor
\State \textbf{return} $\mathsf{fixes}$
\EndProcedure
\end{algorithmic}
\end{algorithm}

Algorithm~\ref{algo:fix} provides the pseudo-code for the resolution of violations of $\fda$.
It first updates the $\dsf{a}$ by using Algorithm~\ref{algo:update-dsf}.
Then, for each class of tuple identifiers in the updated $\dsf{a}$, it gets the tuples from $\repair$ with identifiers in this class and stores these tuples in a variable $\repair_\mathcal{E}$ (line~\ref{line:get-tuples}).
These tuples are inspected for their values for $a$ (line~\ref{line:violation}).
If a violation is observed, the repair function for $a$ is applied to the bag of values $\bproj{\repair_{\mathcal{E}}}{a}$ and the result is stored in the variable $v_\mathsf{fix}$ (line~\ref{line:repair-function}).
This value is then assigned to attribute $a$ for all the tuples of $\repairi{i-1}$ involved in the violation (line~\ref{line:fix})
\begin{example}
In Example~\ref{ex:update-dsf} we showed that $\dsf{\mathsf{\#provider}}$ has two equivalence classes.
Suppose we choose as repair function majority voting with random tie breaking then fixing $\fdxyt{hospital\ name}{\#provider}$ must change the value of $\mathsf{\#provider}$ for tuple $4$ into $\mathsf{10006}$.
In addition, we must change also tuple $5$ or $6$.
\end{example}

\paragraph{Priority repair} We now have all components in place to present an algorithm that performs priority repair on a class of attributes $\cl{i}$.
The pseudo-code for this algorithm is shown in Algorithm~\ref{algo:priority-repair}.
The algorithm receives a class of attributes $\cl{i}$ and the corresponding set of FDs $\fds_i :=\fds[\schemai{i}] \setminus \fds[\schemai{i-1}]$, which contains only those FDs that are possibly still violated in $\repairi{i-1}$.
The algorithm starts with sorting the FDs according to the principles explained before: we prioritize pilot FDs (see Eq.~\eqref{eq:pilot-fd}) followed by the other FDs, which are sorted by their right-hand side attribute based on the cardinality of the set $\mathsf{Vio}(.)$.
FDs are added to a stack that maintains this order \emph{at all time} (line~\ref{line:sort-fds}).
\begin{algorithm}[!htb]
\caption{Repair FDs $\fds_i$ by changing attributes $\cl{i}$ to obtain partial repair $\repairi{i}$}
\label{algo:priority-repair}
\begin{algorithmic}[1]

\Require{$\cl{i}$ is a class of $\parti$ such that $\fseqrep{\parti}{\fds}$}
\Procedure{priorityRepair}{$\repairi{i-1}, \fds_i, \cl{i}$}
\State $\textbf{init}\left(\mathbb{S}, \fds_i\right)$ \Comment{Sorted stack of FDs} \label{line:sort-fds}
\State $\forall a\in\cl{i}: \textbf{init}\left(\dsf{a}\right)$ \Comment{Initialize DSF with singleton classes}\label{line:init-dsf}
\While{$\mathbb{S}\neq\emptyset$}
    \State $\fd\gets\textbf{poll}\left(\mathbb{S}\right)$
    \State $\mathsf{fixes}\gets\textbf{fix}\left(\repairi{i-1}, \fd, \dsf{\rhs{\fd}}\right)$ \Comment{Fix step} \label{line:fix-violations}
    \If{$\mathsf{fixes}\neq 0$} \Comment{Revision step}
        \For{$\fd'\in\left(\fds_i\setminus \mathbb{S}\right)$}
            \If{$\rhs{\fd} \in \lhs{\fd'}$} \label{line:revision}
                \State $\textbf{add}\left(\mathbb{S}, \fd'\right)$ 
            \EndIf
        \EndFor
    \EndIf 
\EndWhile
\State \textbf{return} $\repairi{i-1}$
\EndProcedure
\end{algorithmic}
\end{algorithm}

Next, we initialize a DSF structure for each $a\in\cl{i}$.
This initialization creates a singleton class for each $r\in\repairi{i-1}$ and thus requires $|\cl{i}|\cdot|\repairi{i-1}|$ $\mathsf{makeset}$ operations (line~\ref{line:init-dsf}).
We then start a loop that continues until the stack with FDs is empty.
In each iteration, we poll the first FD from the stack and perform two steps: the fix step and the revision step.
In the \emph{fix} step, we retrieve the current DSF structure for $\rhs{\fd}$ and we apply Algorithm~\ref{algo:fix} to fix any violations (line~\ref{line:fix-violations}).
In the \emph{revision} step, we verify if the fix step made changes to $\repairi{i-1}$. 
If this is the case, then the values for $\rhs{\fd}$ in $\repairi{i-1}$ have been modified.
We then check if there are FDs $\fd'$ that have $\rhs{\fd}$ in the left-hand side and add those FDs back to the stack to check if new violations have been introduced (line~\ref{line:revision}).
This continues until no new violations are found.

Algorithm~\ref{algo:priority-repair} can now be attributed with the following properties.
First of all, for each attribute $a$, the DSF structure is initialized only once and during each update (Algorithm~\ref{algo:update-dsf}), we change $\dsf{a}$ only by merging classes.
In particular, this means that $\dsf{a}$ \emph{before} update is a partition refinement of $\dsf{a}$ \emph{after} update.
As such, any FDs that were previously repaired and have $a$ in the right-hand side, remain satisfied.
For FDs that have $a$ in the left-hand side, this might not be the case and these FDs are added to the stack for revision (only if $\mathsf{fixes}>0$).
During such a revision, $\dsf{a}$ can again only change by merging classes, which means that in consecutive revisions, the number of classes in $\dsf{a}$ is monotonically decreasing.
This leads us to the formulation of the following theorem.
\begin{theorem}
\label{theorem:termination}
Let $\parti=[\cl{1}\ldots\cl{m}]$ be a partition of $\schema$ such that $\fseqrep{\parti}{\fds}$ for a set of FDs $\fds$.
Then for any class $\cl{i}$ from $\parti$ and any partial repair $\repairi{i-1}$, we have (i) execution of $\mathbf{priorityRepair}\left(\repairi{i-1}, \fds_i, \cl{i}\right)$ (Algorithm~\ref{algo:priority-repair}) terminates and (ii) after termination $\repairi{i-1}$ is modified into a partial repair $\repairi{i}$.
\end{theorem}
Theorem~\ref{theorem:termination} states that even in case of revisions, we eventually end up with a partial repair $\repairi{i}$.
Nevertheless, the presence of revisions comes with a cost as it requires a new update of the DSF structure and potential modifications to the repair.
It is therefore clear that avoiding revisions has a computational advantage.
In that regard, we note that whenever $|\cl{i}|=1$, all FDs are necessarily pilot FDs and thus never require revision.
Moreover, in the case of unary FDs, revision is not needed if the repair function for the single left-hand side attribute is preservative.
\begin{proposition}
\label{prop:revision}
Let $\parti=[\cl{1}\ldots\cl{m}]$ be a partition of $\schema$ such that $\fseqrep{\parti}{\fds}$ for a set of FDs $\fds$.
Then for any class $\cl{i}$ from $\parti$ and any partial repair $\repairi{i-1}$, a unary FD $\fdxy{a'}{a}$ can be ignored in the revision step of $\mathbf{priorityRepair}\left(\repairi{i-1}, \fds_i, \cl{i}\right)$ if $\rf{a'}$ is preservative.
\end{proposition}
The proposition above shows that if we use preservative repair functions, then repair of unary FDs is extremely efficient as it is revision-free.

\subsection{The Swipe algorithm}
\label{sec:swipe}
The ability to (i) construct a maximally refined partition that is forward repairable and (ii) repair classes of such a partition by means of a priority model can now be used to construct a repair algorithm for FDs that we call the Swipe algorithm.
The pseudo-code to repair a dirty relation $\rel$ in this way is shown in Algorithm~\ref{algo:swipe}.
The algorithm starts by computing a minimal cover of FDs in order to remove any redundant information that is present (line~\ref{line:minimal-cover}).
Based on such a minimal cover, we construct a partition $\parti=[\cl{1}\ldots\cl{m}]$ induced by the preorder $P^+$ (line~\ref{line:build-partition}) as explained in Section~\ref{sec:partition-building}.
We initialize the repair with the original dirty relation $\rel$.
This repair is then modified step by step.
More precisely, we iterate over the classes $\cl{i}$ of $\parti$ in the natural order and in each step apply the priority repair algorithm, where we give the current partial repair $\repairi{i-1}$ as input, together with those FDs that involve at least one attribute from $\cl{i}$ (line~\ref{line:priority-repair}).
During the execution of the priority repair algorithm, $\repairi{i-1}$ is modified by assigning values to attributes in $\cl{i}$ in such a way we get a new partial repair $\repairi{i}$.
Because the partition is forward repairable, we know that such a modification is possible in each iteration.
Eventually, we obtain a repair $\repair:=\repairi{m}$ that satisfies all FDs.
\begin{algorithm}[!htb]
\caption{Swipe}
\label{algo:swipe}
\begin{algorithmic}[1]
\Procedure{swipe}{$\rel,\fds$}
\State $\fdcov \gets \textbf{minimalCover}\left(\fds\right)$\Comment{Minimal cover}\label{line:minimal-cover}
\State $\parti \gets \textbf{buildFromPreorder}\left(\fdcov \right)$ \Comment{Construct $\parti$ from preorder $P^+$ on $\fds$}\label{line:build-partition}
\State $\repairi{0} \gets \rel$\Comment{Initialize repair}
\For{$i\in \{1,\ldots,|\parti|\}$}
    \State $\repairi{i} \gets \textbf{priorityRepair}(\repairi{i-1}, \fdcov [\schemai{i}] \setminus \fdcov [\schemai{i-1}], \cl{i})$ \Comment{Priority repair on $\cl{i}$ for new FDs} \label{line:priority-repair}
\EndFor
\State \Return $\repairi{m}$
\EndProcedure
\end{algorithmic}
\end{algorithm}

\section{Experimental evaluation}
\label{sec:experiments}
In this section, we report the results of an empirical study that scrutinizes the efficiency and accuracy of the Swipe algorithm.
Because the intent of the paper is to optimize the trade-off between repair quality and computational cost in a Chase-like procedure, we are primarily interested in comparing the Swipe algorithm (a single-sequence Chase procedure) with Llunatic (a multi-sequence Chase procedure).
The goal is hereby to seek evidence for the hypothesis that the Swipe algorithm allows to find good repairs of FD violations in an efficient manner and that exploring many repairs does not necessarily improve the quality of the final repair.
To that end, the study assesses both the quality of repairs in terms of precision, recall and $F$-score (Section~\ref{sec:precision-recall}) as well as the scalability of the Swipe algorithm (Section~\ref{sec:scalability}).
All experiments reported below were executed on a machine running Ubuntu 20.04.1 with 64GB of RAM memory, 12 cores (i9-10920X/3.50GHz) and two 500 GB Solid State Disks in mirror (WDS500G3XHC-00SJG0).

\subsection{Data sets}
\label{sec:exp-data}
\paragraph{Real-life data sets} The study involves four real-life data sets, which differ in size, number of errors and number of considered constraints.
Two of them (Hospital \cite{Dallachiesa2013,Rekatsinas2017,Geerts2019} and Flight \cite{Li2013,Geerts2019,BronselaerAcosta2023}) are commonly used in experiments on data quality, while the others have been composed and used in previous work by the authors.
Each data set is accompanied by a gold standard that contains the correct tuples for a sample of the data set.
\begin{table}[!ht]
\caption{Summary of data sets used in the experiments. Repairable attributes are attributes that occur in at least one FD. Cells in error are cells for which the value is different in $\rel_{\mathsf{gold}}$ compared to $\rel$.}
\label{tab:dataset-summary}
\begin{tabular}{lcccc}
\toprule
                                                                    & Hospital & Allergen & Eudract & Flight \\
\midrule
\# rows                                                             & $1000$     & $1160$     & $86670$   & $776067$ \\ 
\# rows w. ground truth & $1000$     & $206$      & $3133$    & $70951$ \\ 
\# attributes                                                       & $19$       & $22$       & 9       & 5      \\ 
\begin{tabular}[c]{@{}l@{}}\# repairable\\ attributes\end{tabular}  & $15$       & $22$       & 9       & 5      \\ 
\# FDs                                                              & $15$       & $21$       & 11      & 4      \\ 
\# cells in error                                                   & $509$ ($2.68\%$)     & $358$ ($8.28\%$)      & $2962$ ($11.82\%$)    & $123799$ ($43.62\%$) \\
\# cells with \textsc{null}                        & 0        & 0        & $1722$ ($6.87\%$)   & $59837$ ($21.08\%$)  \\
\bottomrule
\end{tabular}
\end{table}

Table~\ref{tab:dataset-summary} provides a summary of the data set descriptions, including the size of the gold standard and the number of constraints considered in our experiments. 
More details on the origin of these data sets and their gold standard, as well as full downloads, are openly available\footnote{\label{footnote:reproducibility}\url{https://gitlab.com/antoonbronselaer/swipe-reproducibility}}.
A short description for each data set is included below.
\begin{itemize}
\item \textbf{Hospital.} A benchmark data set often used in literature on data cleaning \cite{Dallachiesa2013,Rekatsinas2017}.
The original data \cite{Dallachiesa2013} stems from the U.S. Department of Health and
Human Services\footnote{\url{http://www.hospitalcompare.hhs.gov}}.
We use the version consisting of 1000 tuples that has been used before in data repairing experiments \cite{Rekatsinas2017}.
The running example from Figure~\ref{fig:example} is a sample from this data set.
\item \textbf{Allergen.} A data set containing data regarding allergens of food products as found on two different websites \cite{BronselaerAcosta2023}.
The attributes indicate the presence (`2'), traces (`1'), or absence (`0') of allergens in a product which can be identified by its barcode.
A gold standard was compiled by manual consultation of pictures of a sample of products \cite{BronselaerAcosta2023}.
\item  \textbf{Eudract.} A data set with data about clinical trials conducted in Europe \cite{Bronselaer2022,BronselaerAcosta2023} and obtained from the Eudract register\footnote{\url{https://eudract.ema.europa.eu/}}.
A gold standard was compiled based on the information of the same studies in other registries.

\item \textbf{Flight.}  A data set that describes flights annotated with the departing and arrival airports as well as their expected and actual time of departure and arrival.
The original data set was used in the context of data fusion \cite{Li2013,Dong2009,Dong2012} and is available online\footnote{\url{http://lunadong.com/fusionDataSets.htm}}.
In the current study, we use the version and gold standard proposed in recent work \cite{BronselaerAcosta2023}.
\end{itemize}

\paragraph{Artificial data sets} In order to assess scalability, artificial data sets (i.e., an artificial relation $\rel$ paired with a set of constraints $\fds$) are used in Section~\ref{sec:scalability}.
The process to generate these data sets takes two parameters: the number of attributes (i.e., $|\schema|$) and the number of tuples (i.e., $|\rel|$).
For a given number of tuples and attributes, we generate tuples by randomly assigning values to attributes under the condition that $|A|=10$ for each attribute $a$.
By choosing the domain size sufficiently small, we ensure that there are many violations of FDs that need repairing.
Clearly, if we would increase domain sizes $|A|$ we would obtain less violations and observe faster run times.
Once relations are generated, we generated a set $\fds$ such that $|\fds| = |\schema|$, meaning that the number of FDs is equal to the number of attributes.
This reflects the fact that for larger schemata, one can expect more dependencies that must be satisfied.
To generate an FD, we first choose the size of the left-hand side by sampling from a uniform distribution $\mathsf{unif}\{1,\lceil|\schema|/10\rceil\}$.
This reflects that for larger schemata, it becomes more likely to have FDs with multiple attributes in the left-hand side.
We then sample for that FD the required number of attributes uniformly from $\schema$ to compose the left-hand side.
The right-hand side attribute is chosen from the remaining attributes.
When a dirty relation $\rel$ and a set of FDs has been generated, we apply Swipe to obtain a repair $\repair$.

\subsection{Repair techniques}

We compare the Swipe algorithm with Llunatic~\cite{Geerts2019}. %and HoloClean~\cite{Rekatsinas2017}.
The data sets were stored in a \href{https://www.postgresql.org/}{PostgreSQL} database in separate schemas.
Note that we used PostgreSQL version 9.5 because Llunatic relies on the \texttt{``with oids''} option during table creation, which is abandoned as of version 12\footnote{See \url{https://www.postgresql.org/docs/12/release-12.html}, section E.13.2}.

\paragraph{Llunatic.} Llunatic is a general purpose Chase engine that builds a Chase tree where nodes correspond to (sequences of) repair steps and leaf nodes represent repairs.
The implementation allows configuration of (i) pruning strategies to keep the search efficient and (ii) a partial order that indicates preferences of values for repair.
We use the most recent version (2.0) available at \href{https://github.com/donatellosantoro/Llunatic}{GitHub}\footnote{\url{https://github.com/donatellosantoro/Llunatic}}.
We made a local build of the Llunatic code in order to be able to run it from the command line on an \texttt{lxc} container.
We encoded the available FDs for each data set in separate configuration files.
As Llunatic allows for a wide range of configuration options, we decided to create two configurations for each data set:
\begin{itemize}
\item The first configuration is the standard one (S) and is the same for each data set.
It searches for forward repairs only and we set both the branching threshold and the potential solutions threshold to $2$.
This means that the Chase tree is binary and we stop searching after two repairs have been found.
The partial order to resolve conflicts is based on the frequency of values.
The main idea of this standard configuration is that it resembles the search strategy of Swipe.
More specifically, it resolves conflicts in the same way as the default method in Swipe (i.e., majority voting), it uses forward repairing only and it constructs few sequences.
The main difference is thus that Swipe will construct a specific order on the FDs and uses no variables.

\item The second configuration is a fine-tuned configuration (FT) that is tailored to each data set individually.
First of all, we allowed backward repairing for Hospital, Eudract and Flight to see if this helps in finding better repairs.
We did not allow this for Allergen as there we wanted to focus on mirroring the behaviour of Swipe when using the $\mathsf{max}$ repair function.
Second, to account for typographical errors that are present in the attributes of Hospital and Flight, we used a cost manager that uses Smith-Waterman string similarity between values to guide resolution.
For Hospital, Eudract and Flight, we also increased the branching threshold such that a larger part of the search space is explored and more repairs are searched for.
\end{itemize}
A summary of the main differences of the FT configurations is given for each data set in Table~\ref{tab:configurations}.
The full configuration files are available online and can be used for reproduction purposes \footref{footnote:reproducibility}.
\begin{table}[!t]
\caption{Configuration details of LLunatic (FT) for each data set.}
\label{tab:configurations}
\centering
\begin{tabular}{lc}
\toprule
Data set  & Llunatic (FT) \\
\midrule
Hospital & \begin{tabular}[c]{@{}c@{}} Similarity cost manager \\ Higher branching threshold\\ More potential solutions \\Backward repairs allowed\end{tabular}\\ \hline
Allergen & Mimic $\mathsf{max}$ repair function\\ \hline
Eudract  & \begin{tabular}[c]{@{}c@{}}Higher branching threshold\\ More potential solutions \\ Backward repairs allowed\end{tabular}\\ \hline
Flight   & \begin{tabular}[c]{@{}c@{}} Similarity cost manager \\Higher branching threshold\\ More potential solutions \\ Backward repairs allowed\end{tabular}\\ \bottomrule
\end{tabular}
\end{table}
% \textbf{HoloClean} uses a probabilistic model of the data that is learned from available integrity constraints, external data and quantitative statistics about the data.
% It uses probabilistic inference to generate repairs of a given dirty data set.
% We choose HoloClean as a state-of-the-art representative of the learning-based approaches.
% We use the most recent version (1.0) available at \href{https://github.com/HoloClean/holoclean}{GitHub}\footnote{\url{https://github.com/HoloClean/holoclean}}.
% For each data set, we encoded the available FDs as denial constraints and ran HoloClean using the NullDetector and the ViolationDetector jointly.
% Any other parameters were set to their default.
\paragraph{Swipe.} The Swipe algorithm was implemented in Java (Version 8 or higher) as part of the $\mathsf{ledc}$ framework\footnote{\href{https://gitlab.com/ledc/ledc-fundy}{https://gitlab.com/ledc/ledc-fundy}}, which is an open-source framework that bundles various techniques for data quality under one umbrella project.
The implementation of Swipe contains a set of basic repair functions.
We consider three of them in our experiments:
\begin{itemize}
\item The first one is a \emph{majority} voting function ($\mathsf{mv}$) with random tie breaking.
This repair function was used in the examples throughout this paper and is considered a default choice.
\item The second one is a \emph{weighted} voting function ($\mathsf{wv}$) with random tie breaking.
The weight of a value in the bag is based on the number of \textsc{null} values that occur in the row where that value is taken.
More specifically, if a value from the bag is taken from a tuple with $N$ \textsc{null} values, the weight for that value is equal to $\left(|\schema| - N\right)^4$.
We include this repair function as it has been shown effective in previous work \cite{BronselaerAcosta2023}.
\item The third repair function simply chooses the \emph{maximal} value ($\mathsf{max}$) and can be used in a scenario where attribute domains are equipped with an order relation that is relevant to the repair.
This is true for all datasets included in the experiment, except for the Hospital dataset, where using the natural order simply sorts the attribute values alphabetically.
\end{itemize}
For each data set, we run Swipe three times and in each run use a different repair function.
That repair function is then used for \emph{all} attributes in the data set.
Note that each of these three functions is \emph{preservative}, which means Proposition~\ref{prop:revision} applies in each scenario we test here.

\subsection{Precision and recall}
\label{sec:precision-recall}
In order to examine the quality of repairs, we compare, for each data set, the given relation $\rel$ with a repair $\repair$ and a gold standard $\rel_{\mathsf{gold}}$ and compute the following measures:
\begin{itemize}
    \item Precision ($P$): the number of correctly repaired cells divided by the number of repaired cells.
    \item Recall ($R$): the number of correctly repaired cells divided by the number of erroneous cells.
    \item $F$-score ($F$): harmonic mean of precision and recall.
\end{itemize}
Hereby, a \emph{repaired} cell is a cell for which the value is different in $\repair$ as compared to $\rel$.
A \emph{correctly repaired} cell is a repaired cell where in addition the values are the same in $\repair$ and $\rel_{\mathsf{gold}}$.
An \emph{erroneous} cell is a cell for which the value is different in $\rel_{\mathsf{gold}}$ compared to $\rel$.
As Llunatic uses a multi-sequence Chase method, it provides multiple repairs.
We report here the results for the repair that has the lowest cost.
Moreover, in the case of Llunatic, $\repair$ may contain named variables ($\mathsf{llun}$-values) that indicate that the cell should be assigned some constant value, but the value self remains unspecified.
Llunatic offers the option to call on human intervention to fill in the correct constant.
To report repair quality in the presence of variables, we propose to consider two extreme cases: one in which all variables are assigned the \emph{correct} constant and one in which all variables are assigned an \emph{erroneous} constant.
These cases provide a best and worst outcome of precision, recall and $F$-value whenever all variables would be assigned with constant values at random.
As such, for Llunatic, we will report \emph{intervals} of repair quality rather than single values.
Recent studies on Llunatic have used measures that correspond to our proposed upper bound \cite{Geerts2019}.
However, the inclusion of the lower bound introduces an important dimension to the analysis of results in the sense that a large amount of variables implies a big gap between lower and upper bound.
This is important to recognize as many variables imply many human interventions after the repair has been produced.

\begin{table}[!tbh]
\caption{Overview of precision ($P$), recall ($R$) and $F$-score obtained for Llunatic (two configurations) and Swipe (three repair functions).
Best results are marked in bold font.
n/a indicates the repair function was not applicable on that dataset.
}
\label{tab:repair-quality}
\centering
\begin{tabular}{ccccccc}
\toprule
                     & Data set  & Llunatic (S) & Llunatic (FT) & Swipe ($\mathsf{mv}$) & Swipe ($\mathsf{wv}$) & Swipe ($\mathsf{max}$)       \\
\midrule
\multirow{4}{*}{$P$} & Hospital & [0.93,0.96]  & [0.11,0.12]    & \textbf{0.96}  & \textbf{0.96} & n/a  \\ 
                     & Allergen & [0.00,0.57]  & \textbf{[0.64,0.64]} &  0.21 & 0.21 & \textbf{0.64} \\  
                     & Eudract  & [0.80,0.84]  & [0.80,0.84]    & 0.83 & \textbf{0.84} & 0.23\\  
                     & Flight   & [0.67,0.72]  & \textbf{[0.80,0.80]} & 0.71 & 0.79 &  0.56\\
\midrule
\multirow{4}{*}{$R$} & Hospital & [0.83,0.85]  & [0.29,0.31]   & \textbf{0.89} & \textbf{0.89} & n/a \\ 
                     & Allergen & [0.00,0.54]  & [0.30,0.30]   & 0.10 & 0.10 & 0.30  \\
                     & Eudract  & [0.31,0.32]  & [0.31,0.32]   & 0.36 & \textbf{0.37} & 0.32 \\ 
                     & Flight   & [0.66,0.71]  & [0.54,0.54]   & 0.68 & \textbf{0.79} & 0.58 \\
\midrule
\multirow{4}{*}{$F$} & Hospital & [0.88,0.90]  & [0.16,0.17]   & \textbf{0.92} & \textbf{0.92} & n/a \\ 
                     & Allergen & [0.00,0.55]  & [0.41,0.41]   & 0.13 & 0.13 & 0.41 \\  
                     & Eudract  & [0.44,0.47]  & [0.44,0.47]   & 0.51 & \textbf{0.52} & 0.27\\ 
                     & Flight   & [0.67,0.72]  & [0.65,0.65]   & 0.69 & \textbf{0.79} & 0.57\\
\bottomrule
\end{tabular}
\end{table}
% For HoloClean, we always use the recommended parameter settings as communicated by the authors on the GitHub page.
% In addition, we had to make slight modifications to the data sets in order to produce (correct) results.
% First, we used a modified version of the Flight data set where dates were encoded as $4$-byte integers by computing the number of minutes between each date and a fixed point in time.
% This yielded a more compact representation of the data set and was done to prevent out-of-memory errors during computation of the repair.
% Second, the learning model of HoloClean uses co-occurrence statistics of the attributes in the data set.
% In this part of the process, HoloClean turns out to be \emph{very} sensitive to attributes that are uncorrelated to the attributes that require repair.
% Examples of such attributes in our data sets include identifiers and source information.
% The results reported below are those achieved after removing all attributes not required during repair from the data sets.
Table~\ref{tab:repair-quality} provides an overview of the repair quality obtained by the two configurations of Llunatic and the three configurations of Swipe on all data sets.
For each combination of data set and measure of quality, we indicate the best performing approach in bold font.
For recall and $F$-score on the Allergen data set, we did not indicate a best approach as it cannot be indicated with certainty.
When inspecting the results in Table~\ref{tab:repair-quality}, several interesting observations can be made.
%First of all, it can be seen that HoloClean (with recommended parameter settings) scores very high at precision but remarkably lower on recall.
%It is almost always (in three of four cases) the best tool in terms of precision, but the $F$-scores indicate that the loss in recall is significant.
A first observation is that results are sensitive to the exact configurations.
This is especially true for the Allergen dataset, which provides allergen information for each product coming from only two sources.
This low number of sources makes voting procedures less suitable.
The resolution of conflicts by means of the $\mathsf{max}$ repair function, is a superior choice here.
This behaviour is also encoded in the FT scenario of Llunatic and it was confirmed that in this case, Llunatic and Swipe produced exactly the same results.

Interestingly, for other datasets, the voting scenarios of Swipe provide robust and good results.
When a data set has no \textsc{null} values (Hospital and Allergen), both voting scenarios are equivalent.
When \textsc{null} values are present (Eudract and Flight), weighted voting produces better results and turns out to be a good baseline choice.
The $\mathsf{max}$ function provide good results only in particular scenarios like the Allergen dataset.

For Llunatic, the sensitivity to configuration choices seems to be higher than for Swipe.
For the Hospital data set, there is a large difference in repair quality between the two configurations, most likely due to the usage of backward repairing in the second configuration.
In general, we found that giving Llunatic the possibility to include backward repairs (this is the case in the fine-tuned configurations) seems to decrease repair quality on the data sets considered here.
Generalizing this observation should be done with caution, but the results reported here make a strong case for at least questioning the usefulness of backward repair on real-life data sets.

A second observation is that when we compare the repair quality of Swipe with that of the configurations of Llunatic, it can be seen that Swipe produces repairs that are sometimes comparable and usually better than the best possible outcomes of Llunatic.
The improvement in $F$-score is hereby mostly attributed to an improvement in recall, whereas precision behaves comparable.
In this regard, it is interesting to see that Llunatic allows the usage of variables (i.e., lluns) when generating repairs. Opposed to that, repair functions used in the Swipe algorithm always produce constants.
The results in Table~\ref{tab:repair-quality} show that the more aggressive strategy where we always choose some constant is usually paying off.
Moreover, it can be seen that the number of variables might become very high.
This happens for example with the Allergen data set when using the standard configuration of Llunatic.
A high number of variables can be problematic as it requires much human intervention in the repair process.
This might however be a necessity in order to obtain good repairs.
Indeed, we see that for the Allergen data set, the upper bound of recall and $F$-score are much higher for the standard configuration of Llunatic than for the other approaches.

An interesting result is found for the Eudract data set, which is the only data set where we encounter a partition class with more than $2$ attributes.
For this class with $3$ attributes, the Swipe algorithm estimates the reliability of the involved attributes and uses that estimate to rank order FDs.
Table~\ref{tab:repair-quality} shows that using this order of FDs eventually produces a repair with $P=0.84$ and $R=0.37$.
To investigate how good this order is, we repeated the experiment but forced Swipe to use a different priority model for this class.
More precisely, we considered every possible priority order for the involved attributes in order to compute a best and worst case for precision and recall.
We then found that precision ranged between $0.69$ and $0.84$ and recall ranged between $0.31$ and $0.37$.
These additional results shows that the order in which we repair FDs can indeed make a large difference in the repair quality.
However, the sequence based on estimated reliability leads to the best possible repair quality.
This is true for the Eudract data set and also for the Hospital data set, where we have one partition class with $2$ attributes.
Also in the latter case, it is confirmed that the chosen order of FDs is the best one, although the difference in quality is less pronounced.

\subsection{Scalability}
\label{sec:scalability}
\paragraph{Run time comparison} In this section, we investigate the run time efficiency and scalability of Swipe.
First, we measure the mean run time over five runs of Swipe and Llunatic on the different data sets and report the mean execution times (in seconds).
We do not report other scalability parameters like CPU usage or main memory consumption, although we have investigated these parameters and came to similar conclusions.
In case of Llunatic, we differentiate between the standard configuration and the configuration that was fine-tuned for each data set (Table~\ref{tab:configurations}).
For Swipe, we took for each dataset, the best performing repair function.
%In case of HoloClean, we report the total execution time as well as the time spent on the actual repairing.
The results are reported in Table~\ref{tab:execution-times}.
% \begin{table}[!t]
% \centering
% \caption{Average execution times (in seconds) of the different repair methods.}
% \label{tab:execution-times}
% \begin{tabular}{lcccc}
% \toprule
% Data set  & Swipe & \begin{tabular}[c]{@{}c@{}}HoloClean\\ Total (Repair)\end{tabular} & \begin{tabular}[c]{@{}c@{}}Llunatic \\ (S) \end{tabular} & \begin{tabular}[c]{@{}c@{}}Llunatic \\ (FT) \end{tabular}     \\ \midrule
% Hospital & \textbf{0.20}  & 46.46 (38.77)     &  20.31 &  167.09    \\ 
% Allergen & \textbf{0.28}  & 84.35 (74.97)     &  1.52  &  264.82    \\ 
% Eudract  & \textbf{7.15}  & 4463.58 (3228.74) &  19.38 &  98.09     \\ 
% Flight   & \textbf{17.80} & 6807.54 (3520.93) &  910.05 &  3987.90  \\ \bottomrule
% \end{tabular}
% \end{table}

\begin{table}[!t]
\centering
\caption{Mean execution times (in seconds) of the repair methods.}
\label{tab:execution-times}
\begin{tabular}{lccc}
\toprule
Data set  & Llunatic (S) & Llunatic  (FT)   & Swipe (best)   \\
\midrule
Hospital  &  20.31 &  167.09    & \textbf{0.20} \\ 
Allergen  &  1.52  &  264.82    & \textbf{0.28} \\ 
Eudract   &  19.38 &  98.09     & \textbf{7.15} \\ 
Flight    &  910.05 &  3987.90  & \textbf{17.80} \\

\bottomrule
\end{tabular}
\end{table}

The results from Table~\ref{tab:execution-times} show that Swipe outperforms Llunatic on each data set in terms of execution time.
This comes as no surprise as Swipe is designed explicitly as a simplification of the multi-sequence Chase algorithm that underlies Llunatic.
Yet, we observe that the gain in execution time is considerable.
Swipe is two to three orders of magnitude faster than the fine-tuned configuration of Llunatic.
When we compare to the standard Llunatic configuration, the difference reduces to one order of magnitude.
This provides us with some interesting insights.
First of all, the cost manager of Llunatic does what it is supposed to do: restricting repair operations to forward repair only and limiting the branching factor significantly reduces the execution time of Llunatic.
However, Table~\ref{tab:repair-quality} shows that this gain in computational efficiency also leads to lesser quality of the repairs.
In that sense, the Swipe algorithm makes more considerate choices in reducing computational cost with better repair quality as a result.

\paragraph{Size of classes in $\parti$}
The execution time of Swipe is strongly influenced by the size of classes in $\parti$.
For smaller partition classes we usually require less revisions.
To that extent, we computed the sizes of the partition classes for each of the four data sets and came to the conclusion that over all data sets, there are 47 classes of which 45 where singletons.
One class had size 2 (Hospital) and one class had size 3 (Eudract).
This observation confirms that in real-life data sets, it is reasonable to apply the partitioning approach and to separate the repairs of different FDs.

\paragraph{Scalability in terms of $|\rel|$ and $|\schema|$}
In order to gain further understanding of the scalability of Swipe, we experimented with randomly generated data and constraints.
These data and constraints are generated as explained in Section~\ref{sec:exp-data}.
The repair functions are defaulted to majority voting with random tie breaking for all attributes.
We generate repairs of the given dirty data sets with Swipe and then measure the time to generate a repair in milliseconds.
For a fixed number of tuples and fixed number of attributes, we repeat the procedure $10$ times and report the mean of the measured run times.
\begin{figure}[h]
\centering
\resizebox{0.48\columnwidth}{!}{\input{scalability-rows.tikz}}
\resizebox{0.48\columnwidth}{!}{\input{scalability-attributes.tikz}}
\caption{Mean run time (ms) of $10$ executions of Swipe in function of changing $|\rel|$ (left) and $|\schema|$ (right).
}
\label{fig:scalability}
\end{figure}

Figure~\ref{fig:scalability} shows the obtained mean run times for a varying number of tuples (left) and a varying the number of attributes (right).
More precisely, the left panel shows the evolution of the mean run time as the number of tuples increases with powers of 10, starting at $|\rel|=10^2$ and ending at $|\rel|=10^6$.
This trend is shown three times: one time for a small schema ($|\schema|=5$), one time for a medium-sized schema ($|\schema|=25$) and one time for a large schema ($|\schema|=50$).
These results show that the run time of Swipe scales linearly in terms of increasing $|\rel|$ for different schema sizes.
Note in addition that the shift in run time from the small schema ($5$ attributes) to the large schema ($50$ attributes) is two orders of magnitude, suggesting quadratic behavior in terms of $|\schema|$.
This is confirmed in the right panel, that shows the evolution of the mean run time as the number of attributes increases with $5$, starting at $|\schema|=5$ and ending at $|\schema|=50$.
This trend is again shown three times: one time for a small relation ($|\rel|=10^3$), one time for a medium-sized relation ($|\rel|=10^4$) and one time for a large relation ($|\rel|=10^6$).
These results show that Swipe scales quadratic in terms of increasing $|\schema|$ for different relation sizes.
In addition, for an increase in $|\rel|$ with one order of magnitude, we observe shifts of one order of magnitude in run time, confirming the linear behaviour in terms of $|\rel|$.
These results show that Swipe is very efficient in repairing large relations when the number of attributes is relatively small.
For medium-sized schemata and number of FDs, Swipe can still repair relations relatively fast, especially if we compare to other state-of-the art repair methods.

\section{Conclusion}
\label{sec:conclusion}
In this paper, we have introduced the Swipe algorithm to repair violations of functional dependencies (FDs).
This algorithm is a degenerate variant of the Chase-based approach towards FD repairing.
It hinges on two key principles.
The first principle is to use a partition of attributes that is forward repairable.
We have provided an algorithm that, for a given set of FDs, constructs the most refined partition that meets this requirement of forward repairability.
The second principle is that of priority repairing, which fixes the order in which FDs are treated.
We have shown a simple heuristic to build such a priority model based on the estimated reliability of attributes.
From a theoretical point of view, we have shown that Swipe is ensured to terminate and that there are easy to meet conditions under which unary FDs are revision-free.
Empirical results show that Swipe provides an excellent trade-off between repair quality and computational efficiency.
Future improvements of Swipe can focus on applying the principles developed here, to more expressive constraints like conditional functional dependencies.

\bibliographystyle{acm}
\bibliography{swipe-references}

\appendix
\section{Proofs}
\label{appendix-proofs}

\begin{proof}[Proof of Proposition~\ref{prop:seq-rep-implication}]
If $\parti$ satisfies the condition from Definition~\ref{def:forward-seq-repair}, we can assign, for each $\cl{i}\in\parti$ with partial repair $\repairi{i-1}$ and for each $\fd\in\fds[\schemai{i}]\setminus\fds[\schemai{i-1}]$, equal values to $\rhs{\fd}$ in order to satisfy $\fd$.
\end{proof}

\begin{proof}[Proof of Theorem~\ref{theorem:partition}]
We show first that $\parti$ allows sequential forward repairability for $\fds$ (i) and then that $\parti$ cannot be refined without losing this property (ii).

\fbox{(i)} If $\parti$ is induced by $P^+$ for a given set $\fds$, then consider an arbitrary class $\cl{i}$ from $\parti$ and any FD $\fd\in\fdcov[\schemai{i}]\setminus\fdcov[\schemai{i-1}]$.
If $\rhs{\fd}\notin\cl{i}$, then $\parti$ cannot be induced by $P^+$ because it would violate either Eq.~\eqref{eq:equiv-criterion} or Eq.~\eqref{eq:sort-criterion}.
It follows that we must have $\rhs{\fd}\in\cl{i}$ and thus $\fseqrep{\parti}{\fds}$.

\fbox{(ii)} Suppose there is a refinement $\parti'$ of $\parti$ for which $\fseqrep{\parti'}{\fds}$.
This would imply that there is a class from $\parti$ that can be split into two disjoint subclasses without breaking forward repairability.
In turn, this would imply that there is an equivalence class from $\equiv_{P^+}$ that can be split into two subclasses.
Let us denote the equivalence relation in which this split is done by $\equiv_{P'}$.
Now, if $\equiv_{P'}$ still contains $P$ then $P^+$ cannot be the transitive closure of $P$ and we have a contradiction.
Alternatively, if $\equiv_{P'}$ does not contain $P$, then there is at least one $(b,a)\in P$ that is not accounted for and $\parti'$ is not forward repairable.
In both cases, we obtain a conclusion that is in contradiction with the premise and it follows no refinement of $\parti$ is forward repairable.
\end{proof}

\begin{proof}[Proof of Theorem~\ref{theorem:termination}]
We first prove that Algorithm~\ref{algo:priority-repair} terminates (i) and then that after termination we have obtained a partial repair $\repairi{i}$ (ii).

\fbox{(i)} To see that $\mathbf{priorityRepair}\left(\repairi{i-1}, \fds_i, \cl{i}\right)$ terminates, note that each time an FD $\fd$ is considered, either $\dsf{\rhs{\fd}}$ does not change or its number of classes decreases.
If for all $a\in\cl{i}$, $\dsf{a}$ does not change anymore, the algorithm terminates. 
Else, we must reach a point in which each $\dsf{a}$ has only one class.
In that case, for each $a\in\cl{i}$, each tuple in $\repairi{i-1}$ gets the same value for $a$.
It follows that in this case, there are no more violations, from which it follows that all $\dsf{a}$ remain unchanged and the algorithm stops.

\fbox{(ii)} Each time an FD $\fda$ is polled from the stack, the fix step updates $\dsf{a}$ and potentially changes $\repairi{i-1}$ in attribute $a$.
In the revision step that follows, any FD $\fd'\in\fds_i$ that is not on the stack, was polled and repaired before and there are three options.
First, if $a\notin\lhs{\fd'}$ and $a\neq\rhs{\fd'}$, then clearly this FD is still satisfied.
Second, if $a=\rhs{\fd'}$, then because $\dsf{a}$ changes only by merging classes, any two tuples with the same value for $\lhs{\fd'}$ are still in the same class in $\dsf{a}$ and thus will also have the same value after fixing $\fda$.
Hence, $\fd'$ remains satisfied.
Third, if $a\in\lhs{\fd'}$, we put $\fd'$ back on the stack for revision.
It follows that after application of $\textbf{fix}\left(\repairi{i-1}, \fda, \dsf{a}\right)$, any FD from $\fds_i$ not on the stack is currently satisfied by $\repairi{i-1}$.
At the same time, at termination time, the stack is empty and thus all FDs are satisfied.
\end{proof}

\begin{proof}[Proof of Proposition~\ref{prop:revision}]
For an FD $\fdxy{a'}{a}$ during a revision step, we can make a sequence of observations:
\begin{enumerate}
\item If $a'\notin\cl{i}$ then we have $\forall \fd\in\fds_i: \rhs{\fd}\neq a'$ and consequently $\fdxy{a'}{a}$ is never considered during any revision step.
We therefore assume that $a'\in\cl{i}$.

\item $\fdxy{a'}{a}$ is considered in a revision step only if it is not an element of $\mathbb{S}$.
This means it was fixed already and thus that any two tuples with equal values for $a'$ are in the same class of $\dsf{a}$.

\item $\fdxy{a'}{a}$ is considered in a revision step if the preceding fix step involves an FD $\fdxy{X}{a'}$ and is therefore of the form
 $\mathbf{fix}\left(\repairi{i-1}, \fdxy{X}{a'}, \dsf{a'}\right)$.
In the remainder of this proof, we denote the partial repair \emph{before} this fix step as $\rbefore$ and \emph{after} this fix step as $\rafter$.
Clearly, $\rafter\models \fdxy{X}{a'}$.

\item From (1) and (2) we have that $\fdxy{X}{a'}$ mentioned in (3) is not a pilot FD. Thus, FDs are polled from the stack in order induced by $\succ$.

\item The fix step in (3) is preceded by a poll of $\fdxy{X}{a'}$, so (2) implies $a\succ a'$.

\item From (2) and (5) we have that $\rbefore\models \fdxy{a'}{a}$. In addition, transitivity of FDs implies $\rbefore\models \fdxy{X}{a}$.
\end{enumerate}

Suppose now $\rafter$ fails $\fdxy{a'}{a}$, then there must exist two rows $r_1\in\rbefore$ and $r_2\in\rbefore$ that are transformed into $r^*_1\in \rafter$ and $r^*_2\in\rafter$, respectively, for which we have
\begin{equation*}
r^*_1[a'] = r^*_2[a'] \wedge r^*_1[a] \neq r^*_2[a].    
\end{equation*}
In addition, we know that:
\begin{equation*}
r^*_1[a']\neq r_1[a']\vee r_2[a']\neq r^*_2[a']
\end{equation*}
because otherwise $r_1[a'] = r^*_1[a'] = r^*_2[a'] = r_2[a']$ from which
$\rbefore\models\fdxy{a'}{a}$ that $r^*_1[a] = r^*_2[a]$ and this contradicts our construction of $r_1$ and $r_2$.
We can then distinguish between two cases.

\fbox{1} If $r^*_1[X] = r^*_2[X]$ then we have that $\rafter\not\models\fdxy{X}{a}$ and because $a'\notin X$, we have $\rbefore[X] = \rafter[X]$.
As such, it follows that $\rbefore\not\models\fdxy{X}{a}$, which
contradicts with (6).

\fbox{2} If $r^*_1[X] \neq r^*_2[X]$ then because $\rf{a'}$ is preservative, there must exist rows $r_3\in \rbefore$ and $r_4\in \rbefore$ such that on one hand:
\[r_1[X] = r_3[X] \wedge r^*_1[a'] = r^*_3[a'] = r_3[a']\]
and on the other hand
\[r_2[X] = r_4[X] \wedge r^*_2[a'] = r^*_4[a'] = r_4[a'].\]
That is, $r_1$ and $r_2$ received their values for $a'$ from respectively $r_3$ and $r_4$ and this is only possible if $r_1$ and $r_3$ have the same value for $X$ and $r_2$ and $r_4$ have the same value for $X$.
Since we assumed $r^*_1[a'] = r^*_2[a']$, it follows that $r_3[a'] = r_4[a']$ and because $\rbefore\models \fdxy{a'}{a}$ we must also have $r_3[a] = r_4[a]$.
Finally, since $\rbefore\models \fdxy{X}{a}$, we also find that $r_1[a]=r_3[a]$ and $r_2[a] = r_4[a]$ by which we find that $r_1[a]=r_2[a]$.
Again, this contradicts our construction of $r_1$ and $r_2$.
\end{proof}

\end{document}

%% file: scalability-rows.tikz
\begin{tikzpicture}
    \begin{axis}[
        legend pos={north west},
        ylabel={Mean runtime (ms)},
        xlabel={$|\rel|$},
        symbolic x coords={$10^2$,$10^3$,$10^4$,$10^5$, $10^6$},
        xtick=data,
        ymin=1,
        ymode=log,
        xtick pos=left,
        ytick pos=left,
        bar width = 8pt,
        xticklabel style={name=T\ticknum}
    ]
    \addplot[style={lightgray,mark options={fill=white},mark=*}] coordinates {($10^2$,4) ($10^3$,7) ($10^4$,80) ($10^5$,1153) ($10^6$,12489)};
    \addlegendentry{$|\schema|=5$}
    \addplot[style={darkgray,mark=*, dash pattern=on 1pt off 3pt on 3pt off 3pt}] coordinates {($10^2$,17) ($10^3$,63) ($10^4$,758) ($10^5$,9043) ($10^6$,104955)};
    \addlegendentry{$|\schema|=25$}
    \addplot[style={black,mark=none}] coordinates {($10^2$,54) ($10^3$,398) ($10^4$,4204) ($10^5$,80449) ($10^6$,700736)};
    \addlegendentry{$|\schema|=50$}

    \end{axis}

\end{tikzpicture}

%% file: scalability-attributes.tikz
\begin{tikzpicture}
    \begin{axis}[
        legend pos={north west},
        ylabel={Mean runtime (ms)},
        xlabel={$|\schema|$},
        symbolic x coords={$5$,$10$,$15$,$20$, $25$, $30$, $35$, $40$, $45$, $50$},
        xtick=data,
        ymin=1,
        ymode=log,
        xtick pos=left,
        ytick pos=left,
        bar width = 8pt,
        xticklabel style={name=T\ticknum}
    ]
    \addplot[style={lightgray,mark options={fill=white},mark=*}] coordinates {($5$,7) ($10$,14) ($15$,25) ($20$,41) ($25$,63) ($30$,84) ($35$,118) ($40$,175) ($45$,277) ($50$,398)};
    \addlegendentry{$|\rel|=10^3$}
    \addplot[style={darkgray,mark=*, dash pattern=on 1pt off 3pt on 3pt off 3pt}] coordinates {($5$,80) ($10$,177) ($15$,335) ($20$,539) ($25$,758) ($30$,1128) ($35$,1510) ($40$,2307) ($45$,3202) ($50$,4203)};
    \addlegendentry{$|\rel|=10^4$}
    \addplot[style={black,mark=none}] coordinates {($5$,12489) ($10$,23750) ($15$,43751) ($20$,66874) ($25$,104955) ($30$,138941) ($35$,240694) ($40$,321412) ($45$,356827) ($50$,700736)};
    \addlegendentry{$|\rel|=10^6$}

    \end{axis}

\end{tikzpicture}